\documentclass[11pt]{article}

\pdfoutput=1

\usepackage{geometry}
\usepackage{
	amsmath,
	amsthm,
	thmtools,
	thm-restate
}

\usepackage{amsfonts,amssymb,mathtools}

\usepackage[labelfont=bf, font=small]{caption}

\usepackage{authblk}

\usepackage[hyphens]{url}
\usepackage[hyperindex]{hyperref}
\hypersetup{
	colorlinks=true,
	linkcolor=magenta,
	citecolor=magenta,
    urlcolor=magenta
}

\usepackage[capitalize]{cleveref}

\usepackage{tikz}
\usetikzlibrary{automata,arrows,positioning}

\usepackage{prftree}

\usepackage{enumitem}

\usepackage[
	backend=biber,
	style=ieee,
	url=false
]{biblatex}
\bibliography{main.bib}

\title{Processes Parametrised by an Algebraic Theory\thanks{Schmid, Rozowski, and Silva’s work was partially supported by ERC grant Autoprobe (grant agreement 101002697).}}

\author[1]{Todd Schmid \footnote{\url{todd.schmid.19@ucl.ac.uk}}}
\author[1]{Wojciech Rozowski \footnote{\url{w.rozowski@cs.ucl.ac.uk}}}
\author[2]{Alexandra Silva \footnote{\url{alexandra.silva@cornell.edu}}}
\author[3]{Jurriaan Rot \footnote{\url{jrot@cs.ru.nl}}}

\affil[1]{UCL,	London, UK}
\affil[2]{Cornell University, Ithaca, New York, USA}
\affil[3]{Radboud University, Nijmegen, The Netherlands}

\date{}

\begin{document}

% *** Theorem Environments *** %
\theoremstyle{plain}
\newtheorem{theorem}{Theorem}[section]
\newtheorem{lemma}{Lemma}[section]
\newtheorem{corollary}{Corollary}[section]
\newtheorem{proposition}{Proposition}[section]

\theoremstyle{definition}
\newtheorem{remark}{Remark}[section]
\newtheorem{definition}{Definition}[section]
\newtheorem{example}{Example}[section]
\newtheorem{conjecture}{Conjecture}[section]
\newtheorem{assumption}{Assumption}
\newtheorem{question}{Question}

\tikzset{every state/.style={minimum size=0pt}}
\tikzset{every node/.style={scale=.8}}

\newcommand \At  		{\texttt{At}} 		% Atomic tests
\newcommand \Eq 		{\mathsf{E}}				% Equations
\makeatletter
\newcommand*{\da@rightarrow}{\mathchar"0\hexnumber@\symAMSa 4B }
\newcommand*{\da@leftarrow}{\mathchar"0\hexnumber@\symAMSa 4C }
\newcommand*{\xdashrightarrow}[2][]{%
  \mathrel{%
    \mathpalette{\da@xarrow{#1}{#2}{}\da@rightarrow{\,}{}}{}%
  }%
}
\newcommand*{\da@xarrow}[7]{%
  % #1: below
  % #2: above
  % #3: arrow left
  % #4: arrow right
  % #5: space left 
  % #6: space right
  % #7: math style 
  \sbox0{$\ifx#7\scriptstyle\scriptscriptstyle\else\scriptstyle\fi#5#1#6\m@th$}%
  \sbox2{$\ifx#7\scriptstyle\scriptscriptstyle\else\scriptstyle\fi#5#2#6\m@th$}%
  \sbox4{$#7\dabar@\m@th$}%
  \dimen@=\wd0 %
  \ifdim\wd2 >\dimen@
    \dimen@=\wd2 %   
  \fi
  \count@=2 %
  \def\da@bars{\dabar@\dabar@}%
  \@whiledim\count@\wd4<\dimen@\do{%
    \advance\count@\@ne
    \expandafter\def\expandafter\da@bars\expandafter{%
      \da@bars
      \dabar@ 
    }%
  }%  
  \mathrel{#3}%
  \mathrel{%   
    \mathop{\da@bars}\limits
    \ifx\\#1\\%
    \else
      _{\copy0}%
    \fi
    \ifx\\#2\\%
    \else
      ^{\copy2}%
    \fi
  }%   
  \mathrel{#4}%
}
\makeatother

\newcommand \Exp 		{\mathtt{Exp}} 		% Expressions
\newcommand \Expm 		{{\Exp}/{\equiv}}
\newcommand \SExp 		{\mathtt{SExp}} 		% Expressions
\newcommand \acro[1]	{\(\mathsf{#1}\)} 				% Acronym

\newcommand \Id 		{\text{Id}}			% Identity functor
\newcommand \Sets 		{\mathbf{Sets}}				% Category of sets and functions
\renewcommand \P		{\mathcal{P}_{\omega}}		% Finite powerset functor
\newcommand \M			{\mathcal{M}_{\omega}}		% Finite multiset functor
\newcommand \N 			{\mathbb N}					% Natural numbers
\newcommand \D 			{\mathcal{D}_\omega}		% Probability distribution monad
\newcommand \C			{\mathcal{C}}		% Probability distribution monad

\newcommand \id 		{\operatorname{id}}			% Identity map
\newcommand \supp 		{\operatorname{supp}}		
\newcommand \conv 		{\operatorname{conv}}	
\newcommand \fv 		{\operatorname{fv}}	
\newcommand \bv 		{\operatorname{bv}}

\newcommand \tr[1] 		{\mathrel{\raisebox{-0.1em}{{\footnotesize\(\xrightarrow{#1}\)}}}}
\newcommand \out[1] 	{\mathrel{\raisebox{-0.1em}{{\footnotesize\(\xRightarrow{#1}\)}}}}
\newcommand \trd[1] 	{\mathrel{\raisebox{-0.1em}{{\footnotesize\(\xdashrightarrow{#1}{}\)}}}}

\newcommand \sem[1] 	{\lceil \!\! \lfloor #1 \rceil \!\! \rfloor}
\newcommand \asem[1] 	{\lceil \!\! \lceil #1 \rceil \!\! \rceil}
\newcommand \beh   		{{!}}
\newcommand \gd			{\mathsf{gd}} 				% Guarding operator
\newcommand \skiptt  	{\mathtt{1}} 			% Skip
\newcommand \unit		{\underline{u}}
\newcommand \eff  		{}

\newcommand \note[1]  	{{\color{red} #1}}

\maketitle

\begin{abstract}
	We develop a (co)algebraic framework to study a family of process calculi with monadic branching structures and recursion operators. Our framework features a uniform semantics of process terms and a complete axiomatisation of semantic equivalence. We show that there are uniformly defined fragments of our calculi that capture well-known examples from the literature like regular expressions modulo bisimilarity and guarded Kleene algebra with tests. We also derive new calculi for probabilistic and convex processes with an analogue of Kleene star. 
\end{abstract}

\section{Introduction}
	The theory of processes has a long tradition, notably in the study of concurrency, pioneered by seminal works of Milner~\cite{milner1984complete,milner1980ccs} and many others~\cite{baeten2005history}. 
	In labelled transition systems, a popular model of computation in process theory, processes branch nondeterministically.
	This means that any given action or observation transitions a starting state into any member of a predetermined set of states. 
	In Milner's CCS~\cite{milner1980ccs}, nondeterminism appears as a binary operation that constructs from a pair of programs \(e\) and \(f\) the program \(e + f\) that nondeterministically chooses between executing either \(e\) or \(f\). 
	This acts precisely like the join operation in a semilattice.
	In fact, elements of a free semilattice are exactly sets, as the free semilattice generated by a collection \(X\) is the set \(\P^+ X\) of finite nonempty subsets of \(X\)~\cite{manes1976algebraictheories}. 
	This is our first example of a more general phenomenon: the type of branching in models of process calculi can often be captured with an algebraic theory.
	
	A second example appears in the probabilistic process algebra literature, where the process denoted \(e +_p f\) flips a weighted coin and runs \(e\) with probability \(p\) and \(f\) with probability \(1-p\).
	The properties of \(+_p\) are axiomatised and studied in convex algebra, an often revisited algebraic theory of probability~\cite{andova1999process, swirszcz1973convexity, pumplunR1995convexity, stone1949barycentric}.
	The free convex algebra on a set \(X\) is the set \(\D X\) of finitely supported probability distributions on \(X\)~\cite{stone1949barycentric,bonchiSV2019tracesfor,jacobs2010convexity}. 
	
 	A third example is guarded Kleene algebra with tests (\acro{GKAT}), where the process \(e +_b f\) proceeds with \(e\) if a certain Boolean predicate \(b\) holds and otherwise proceeds with \(f\), emulating the \texttt{if-then-else} constructs of imperative programming languages~\cite{manes1991ifthenelse,bloomT1983ifthenelse,mccarthy1961basis}.
	If the predicates are taken from a finite Boolean algebra \(2^{\At}\), the free algebra of \texttt{if-then-else} clauses on a set \(X\) is the function space \(X^{\At}\).
	This explains why adjacency sets for tree models of \acro{GKAT} programs take the form of functions \(\At \to X\).

	This paper proposes a framework in which these languages can be uniformly described and studied.
	We use the \emph{algebra of regular behaviours} (or \acro{ARB}) introduced in \cite{milner1984complete} as a prototypical example.
	\acro{ARB} employs nondeterministic choice as a branching operation, prefixing of terms by atomic actions, a constant representing deadlock, variables, and a recursion operator for each variable.
	Specifications are interpreted using structural operational semantics in the style of \cite{plotkin2004sos}, which sees the set \(\Exp\) of all process terms as one large labelled transition system.
	This is captured succinctly as a \emph{coalgebra}, in this case a function
	\begin{equation}\label{eq:coalgebra}
		\beta : \Exp \to \mathcal P(V + A\times\Exp)
	\end{equation}
	Only finitely branching processes can be specified in \acro{ARB}, so we will replace \(\mathcal P\) with \(\P\) in (\ref{eq:coalgebra}).
	From a technical point of view, \(\P\) is the monad on the category \(\Sets\) of sets and functions presented by the algebraic theory of semilattices with bottom. 
	
	By substituting the finite powerset functor in (\ref{eq:coalgebra}) with other monads presented by algebraic theories, we obtain a parametrised family of process types that covers the examples above and a general framework for studying the processes of each type.
	Instantiating the framework with an algebraic theory gives a fully expressive specification language for processes and a complete axiomatisation of behavioural equivalence for specifications.
	
	One striking feature of many of the specification languages we construct is that they contain a fragment consisting of nonstandard analogues of regular expressions.
	We call these expressions \emph{star expressions} and the fragment composed of star expressions the \emph{star fragment}. 
	Star fragments extend several existing analogues of basic regular algebra found in the process theory literature, including basic process algebra~\cite{bergstraK1989process} and Andova's probabilistic basic process algebra~\cite{andova1999process},  by adding recursion operators modelled after the Kleene star. 
	
	Milner is the first to notice the star fragment of \acro{ARB} in \cite{milner1984complete}.
	He observes that the algebra of processes denoted by star expressions is more unruly than Kleene's algebra of regular languages, and that it is not clear what the appropriate axiomatisation should be.
	He offers a reasonable candidate based on Salomaa's first axiomatisation of Kleene algebra~\cite{salomaa1966two}, but ultimately leaves completeness as an open problem.
	%\note{
	This problem has been subjected to many years of extensive research~\cite{fokkinkzantema1994basic,fokkink1997perpetual,fokkinkzantema1997termination,baetenCG2006starheight,grabmayerfokkink2020complete,grabmayer2021coinductive}. 
	A potential solution has recently been announced by Clemens Grabmayer and will appear in the upcoming LICS.
	
	Replacing nondeterministic choice with the \texttt{if-then-else} branching structure of \acro{GKAT}, we obtain the process behaviours explored in the recent rethinking of the language~\cite{schmidkappekozensilva2021gkat}.
	This makes the open problem of axiomatising \acro{GKAT} (without the use of extremely powerful axioms like the \emph{Uniqueness Axiom} of \cite{gkat}), stated first in \cite{gkat} and again in \cite{schmidkappekozensilva2021gkat}, yet another problem of axiomatising an algebra of star expressions.
	Our general characterisation of star expressions puts all these languages under one umbrella, and shows how they are derived canonically from a single abstract framework.
	
	In summary, the contributions of this paper are as follows:
	\begin{itemize}
		\item We present a family of process types parametrised by an algebraic theory (\cref{sec:a parameterised family})
		together with a uniform syntax and operational semantics (\cref{sec:specifications of processes}). We show how these can be instantiated to concrete algebraic theories, including guarded semilattices and pointed convex algebras. These provide, respectively, a calculus of processes capturing control flow of simple imperative programs and a calculus of probabilistic processes. 
		\item We define an associated denotational semantics and show that it agrees with the operational semantics (\cref{sec:final coalgebras}). This coincidence result is important in order to prove completeness of the uniform axiomatisation we propose for each process type (\cref{sec:soundness and completeness}). 
		\item Finally, we study the star fragment of our parameterised family and propose
		a sound axiomatisation for this fragment (\cref{sec:star fragments}). We show that star fragments of concrete instances of our calculi yield known examples in the literature, e.g. Guarded Kleene Algebra with tests (\acro{GKAT})~\cite{gkat, schmidkappekozensilva2021gkat} and probabilistic processes of Stark and Smolka~\cite{stark1999complete}. 
	\end{itemize}
	Related work is surveyed in \cref{sec:related work}, and future research directions are discussed in \cref{sec:conclusion}.

	\section{A Parametrised Family of Process Types}\label{sec:a parameterised family}
	In this section, we present a family of process types parametrised by a certain kind of algebraic theory. 
	The processes we care about are stateful, meaning they consist of a set of states and a suitably structured set of transitions between states.
	Stateful systems fit neatly into the general framework of \textit{universal coalgebra}~\cite{rutten2000universal}, which stipulates that the type of structure carried by the transitions can be encoded in an endofunctor on the category \(\Sets\) of sets and functions.  
	Formally, given a functor \(B : \Sets \to \Sets\), a \emph{\(B\)-coalgebra} is a pair \((X,\beta)\) consisting of a set \(X\) of \emph{states} and a \emph{structure} map \(\beta : X \to BX\). 
	A \textit{coalgebra homomorphism} \(h : (X,\beta) \to (Y,\vartheta)\) is a function \(h : X \to Y\) satisfying 
	\(
		\vartheta \circ h = B(h) \circ \beta
	\).
	Many types of processes found in the literature coincide with \(B\)-coalgebras for some \(B\), and so do their homomorphisms. 
	For example, finitely branching labelled transition systems are \(\P(A\times\Id)\)-coalgebras, and deterministic Moore automata are \(O \times \Id^A\)-coalgebras~\cite{rutten1998coinduction}. 
	
	In this paper, we consider coalgebras for functors of the form
	\begin{equation}\label{eq:coalgebra signature}
		B_M := M(V + A \times \Id)
	\end{equation}
	for fixed sets \(V\) and \(A\) and a specific kind of functor \(M :\Sets \to \Sets\).
	Intuitively, there are two layers to the process behaviours we care about: one layer consists of either an \emph{output variable} in \(V\) or an \emph{action} from \(A\) that moves on to another state, and the other layer (encoded by \(M\)) combines output variables and action steps in a structured way.
	
	\begin{example}\label{eg:labelled transition system}
		When \(M = \P\), we obtain Milner's nondeterministic processes~\cite{milner1984complete}.
		Coalgebras for \(B_{\P}\) are functions of the form
		\(
			\beta : X \to \P(V + A \times X)
		\), or labelled transition systems with an additional decoration by variables.
		Write \(x \tr{a} y\) to mean \((a,y) \in \beta(x)\) and \(x \Rightarrow v\) to mean \(v \in \beta(x)\). 
		The image below posits a well-defined \(B_{\P}\)-coalgebra
		
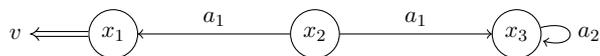
\[\begin{tikzpicture}
			\node (0) {\(v\)};
			\node[state, right=2em of 0] (1) {\(x_1\)};
			\node[state, right=2cm of 1] (2) {\(x_2\)};
			\node[state, right=2cm of 2] (3) {\(x_3\)};
			\draw (2) edge[->] node[above] {\(a_1\)} (1);
			\draw (2) edge[->] node[above] {\(a_1\)} (3);
			\draw[loop right] (3) edge[->] node[right] {\(a_2\)} (3);
			\draw (1) edge[-implies, double, double distance=.5mm] (0);
		\end{tikzpicture}\]
		Its state space is \(\{x_1,x_2,x_2\}\), \(A\) includes \(a_1\) and \(a_2\), and \(v\) is a variable in \(V\). 
	\end{example}

	\paragraph*{Algebraic Theories and Their Monads}\label{sec:algebraic theories}
	We are particularly interested in \(B_M\)-coalgebras when \(M\) is the functor component of a monad \((M,\eta,\mu)\) that is presented by an algebraic theory capturing a type of branching.
	A monad consists of natural transformations \(\eta : \Id \Rightarrow M\) and \(\mu : MM \Rightarrow M\), called the \emph{unit} and \emph{multiplication} respectively, satisfying two laws: 
	\(
	\mu \circ \eta_M = \id_M = \mu \circ M(\eta) \) and \(
	\mu_M \circ \mu = M(\mu) \circ \mu
	\).
	For our purposes, an \emph{algebraic theory} is a pair \((S, \Eq)\) consisting of a polynomial endofunctor \(S = \coprod_{\sigma \in I}\Id^{n_\sigma}\) on \(\Sets\) called an \textit{algebraic signature} and a set \(\Eq\) of equations in the signature \(S\). 
	An element \(\sigma\) of \(I\) should be thought of as an operation with arity \(n_\sigma\). 
	An algebraic theory \((S,\Eq)\) \textit{presents} a monad \((M, \eta, \mu)\) if there is a natural transformation \(\rho : SM \Rightarrow M\) such that for any set \(X\), \((MX,\rho_X)\) is the free \((S,\Eq)\)-algebra on \(X\). 
	That is, \((MX,\rho_X)\) satisfies \(\Eq\) and for any \(S\)-algebra \((Y,\varphi)\) also satisfying \(\Eq\) and any function \(h : X \to Y\), there is a unique \(S\)-algebra homomorphism \(\hat h : (MX, \rho_X) \to (Y,\varphi)\) such that \(h = \hat h \circ \eta\).
	This universal property implies that any two presentations of a given algebraic theory are isomorphic, so we speak simply of ``the'' monad presented by an algebraic theory.

	\begin{example}\label{eg:arb}
		%As mentioned before, Milner's nondeterministic processes are \(B_{\P}\)-coalgebras. 
		The finite powerset functor is part of the monad \((\P, \{-\}, \bigcup)\) that is presented by the theory of semilattices (with bottom).
		The theory of semilattices is the pair \((1+\Id^2, \mathsf{SL})\), since the arity of a constant operation is \(0\) and \(+\) is a binary operation, and \(\mathsf{SL}\) consists of %the equations (\acro{SL1})-(\acro{SL4}) below.
		\[
		x + 0 \stackrel{(\mathsf{SL1})}{=} x
		\qquad 
		x + x \stackrel{(\mathsf{SL2})}{=} x
		\qquad 
		x + y \stackrel{(\mathsf{SL3})}{=} y + x
		\qquad 
		x + (y + z) \stackrel{(\mathsf{SL4})}{=} (x + y) + z
		\]
	\end{example}
	
	Not every algebraic theory has such a familiar presentation as the theory of semilattices, but it is nevertheless true that every algebraic theory presents a monad.
	If we let \(S^*X\) denote the set of \textit{\(S\)-terms}, expressions built from \(X\) and the operations in \(S\), then \((S, \Eq)\) automatically presents the monad \((M, \eta, \mu)\) where \(MX = (S^*X)/\Eq := \{[q]_\Eq \mid q \in S^*X\}\) is the set of \(\Eq\)-congruence classes of \(S\)-terms, \(\eta\) computes congruence classes of variables, and \(\mu\) evaluates terms.
	This is witnessed by letting the transformation \(\rho\) be the restriction of \(\mu\) to the operations of \(S\) on \(S\)-terms.
	We take this to be the default presentation of an arbitrary algebraic theory.
	
	Our aim is to develop a (co)algebraic framework for studying $B_M$-coalgebras when \(M\) is the functor part of a monad presented by an algebraic theory.
	We will make three assumptions about the algebraic theories. First, we rule out the case of \(M\) being the constant \(1\) functor.

	\begin{assumption}
		The theory \(\Eq\) is \emph{nontrivial}, meaning that the equation \(x = y\) is not a consequence of \(\Eq\) for distinct \(x\) and \(y\).
	\end{assumption}

	This is equivalent to requiring that the unit \(\eta\) is injective.
	That is, the \(\Eq\)-congruence classes \([x]_{\Eq}\) and \([y]_\Eq\) in \(MX\) are distinct for distinct variables \(x\) and \(y\) in \(X\). 

	Second, we assume the existence of a constant symbol denoting deadlock, which might occur when recursing on unguarded programs. 

	\begin{assumption}
		Algebraic theories contain a designated constant \(0\). 
	\end{assumption}
	
	Finally, to keep the specifications of processes finite, we make the following assumption despite the fact that it has no bearing on the results presented before \cref{sec:soundness and completeness}. 

	\begin{assumption} \label{asm:finite arities}
		Each operation from S has a finite arity.
	\end{assumption}
	
	We conclude this section with examples of algebraic theories and the monads they present.
	
	\begin{example}\label{eg:guarded semilattices}
		For a fixed finite set \(\At\) of \emph{atomic tests}, the algebraic theory of \emph{guarded semilattices} is the pair \((1 + \coprod_{b \subseteq \At} \Id^2, \mathsf{GS})\), where \(\mathsf{GS}\) consists of the equations
		\[
		x +_b x \stackrel{(\mathsf{GS1})}= x
		\qquad 
		x +_{\At} y \stackrel{(\mathsf{GS2})}= x
		\qquad 
		x +_b y \stackrel{(\mathsf{GS3})}= y +_{\bar b} x
		\qquad 
		(x +_b y) +_c z \stackrel{(\mathsf{GS4})}= x +_{bc} (y +_c z)
		\]
		Here, \(+_b\) is the binary operation associated with the subset \(b \subseteq \At\), \(\bar b := \At \setminus b\), and \(bc := b \cap c\). 
		The theory of guarded semilattices is presented by the monad \(((1+\Id)^{\At}, \lambda \xi.(-), \Delta^*)\), where 
		\(
			(\lambda \xi.x)(\xi) = x
		\)
		and
		\( 
			\Delta^*(F)(\xi) = F(\xi)(\xi)
		\).
		The idea is that \(+_b\) acts like an \texttt{if-then-else} clause in an imperative program.
		This is reflected in a free guarded semilattice \(((1 + X)^{\At}, \rho_X)\), where for a pair of maps \(h_1, h_2 : \At \to X\) we define
		\[
		\rho_X(h_1 +_b h_2)(\xi) := \begin{cases}
			h_1(\xi) &\text{if \(\xi \in b\)} \\
			h_2(\xi) &\text{otherwise}
		\end{cases}
		\]
		The theory of guarded semilattices dates back to the algebras of \texttt{if-then-else} clauses studied in \cite{mccarthy1961basis, manes1991ifthenelse,bloomT1983ifthenelse,bloomE1988varieties}.
		For instance, guarded semilattices are examples of \emph{McCarthy algebras}, introduced by Manes in \cite{manes1991ifthenelse}.\footnote{More information on the theory of guarded semilattices can be found in \cref{sec: guarded semilattices}.}
	\end{example}

	 \begin{example}
	 	By deleting the second axiom of \(\mathsf{SL}\), we obtain the theory of commutative monoids \((\Id^2, \mathsf{CM})\).
	 	This theory presents the finite multiset monad \((\M, \delta_{(-)}, \sum)\), where \[
	 		\M X = \{m : X \to \N \mid \{x \mid m(x) > 0\}~\text{is finite}\}
	 	\]
	 	and \[
	 		\delta_y(x) = [x = y?]
	 		\qquad 
	 		\sum (F)(x) = \sum_{m \in \M X} F(m)\cdot m(x)
	 	\]
	 \end{example}

	\begin{example}\label{eg:pointed convex algebras}
		The theory of \emph{pointed convex algebras} studied in \cite{bonchi2021presenting} is \((1 + \coprod_{p \in [0,1]} \Id^2, \mathsf{CA})\), where \(\mathsf{CA}\) consists of the equations
		\[
		x +_p x \stackrel{(\mathsf{CA1})}= x
		\qquad 
		x +_{1} y \stackrel{(\mathsf{CA2})}= x
		\qquad 
		x +_p y \stackrel{(\mathsf{CA3})}= y +_{\bar p} x
		\qquad 
		(x +_p y) +_q z \stackrel{(\mathsf{CA4})}= x +_{pq} (y +_{\frac{q\bar p}{1-{pq}}} z)
		\]
		Here, \(+_p\) is the binary operation with index \(p \in [0,1]\), \(\bar p := 1 - p\), and \(pq\neq 1\).
		This theory presents the pointed finite subprobability distribution monad \((\D(1+\Id), \delta_{(-)}, \sum)\), 
		where
		\[
			\D(1+X) = \left\{\theta : X \to [0,1] \middle\vert 
			\begin{array}{c}
				\{x \mid \theta(x) > 0\}~\text{is finite} \\
				\sum_{x \in X}\theta(x) \le 1
			\end{array}
			\right\}
		\]
		for any set \(X\), and for any \(x \in X\), \(\theta \in \D(1+X)\), and \(\Theta \in \D(1 + \D(1 + X))\),
		\[
			\delta_x(y) = [x=y?]
			\qquad \qquad
			\sum(\Theta)(\theta) = \sum_{y \in X}\Theta(\theta)\cdot \theta(y)
		\]
		This is witnessed by the transformation \(\rho\) that takes \(0\) to the trivial subdistribution and computes the Minkowski sum
		\(
		\rho_X(\theta +_p \psi) = p \cdot \theta + (1-p)\cdot \psi
		\)
		for each \(p \in [0,1]\), \(\theta,\psi \in \D(1 + X)\).
	\end{example}
	
	\begin{example}\label{eg:pointed convex semilattices}
		The theory of \emph{pointed convex semilattices} studied in \cite{bonchi2021presenting,varacca2006distributing, bonchi_silva_sokolova:2017} combines the theory of semilattices and the theory of convex algebras. 
		It has both a binary operation \(+\) mimicking nondeterministic choice and the probabilistic choice operations \(+_p\) indexed by \(p \in [0,1]\).
		Formally, it is given by the pair \((1 + \Id^2 + \coprod_{p \in [0,1]} \Id^2, \mathsf{CS})\), where \acro{CS} is the union of \acro{SL}, \acro{CA}, and the distributive law  
		\[
			 (x + y) +_p z \stackrel{(\mathsf{D})}= (x +_p z) + (y +_p z)
		\]
		This theory presents the pointed convex powerset monad \((\C, \eta^{\C}, \mu^{\C})\), where \(\C X\) is the set of finitely generated convex subsets of \(\D(1 + X)\) containing \(\delta_0\), and for \(x \in X\) and \(Q \in \C \C X\),
		\[
			\eta^{\C}(x) = \{p\cdot\delta_x \mid p \in [0,1]\} %= \{[x=y?]\}
			\qquad \qquad
			\mu^{\C}(Q) = \bigcup_{\Theta \in Q}\left\{
				\sum_{U \in \C_0X}\Theta(U) \cdot \theta_U~ 
				\middle|~ (\forall U \in \C_0 X)~\theta_U \in U
			\right\}
		\]
		The witnessing transformation \(\rho^\C\) takes \(0\) to \(\{\delta_0\}\), computes the Minkowski sum (extended to subsets) in place of \(+_p\), and interprets the \(+\) operation as the convex union
		\[
			\rho_X^\C(U + V) = \{p \cdot \theta_1 + (1-p) \cdot \theta_2 \mid p \in [0,1], \theta_1 \in U, \theta_2 \in V\}
		\]
	\end{example}

\section{Specifications of Processes}\label{sec:specifications of processes}
	Fix an algebraic theory \((S,\Eq)\) presenting a monad \((M,\eta,\mu)\).
	In this section, we give a syntactic and uniformly defined specification system for \(B_M\)-coalgebras and an associated operational semantics. 
	We are primarily concerned with the specifications of finite processes, and indeed the process terms we construct below denote processes with finitely many states.
	The converse is also true, that every finite \(B_M\)-coalgebra admits a specification in the form of a process term, but we defer this result to \cref{sec:soundness and completeness} because of its relevance to the completeness theorem there.
	 
	The syntax of our specifications consists of variables from an infinite set \(V\), actions from a set \(A\), and operations from \(S\). 
	The set \(\Exp\) of \emph{process terms} is given with the grammar
	\[
		e,e_i ::= 0 \mid v \mid \sigma(e_1, \dots, e_n) \mid a e \mid \mu v ~ e
	\]
	where \(v \in V\), \(a \in A\), and \(\sigma\) is an \(S\)-operation.
	Abstractly, process terms form the initial \(\Sigma_M\)-algebra \((\Exp,\alpha)\), where \(\Sigma_M :\Sets \to \Sets\) is the functor defined by
	\[
	\Sigma_M := S + V + A\times \Id + V\times \Id
	\]
	and the algebra map \(\alpha : \Sigma_M\Exp \to \Exp\) evaluates \(\Sigma_M\)-terms.
	
	\begin{figure}[!t]
		\[
		\begin{aligned}
			\epsilon(v) &= \eta(v) \\
			\epsilon(ae) &= \eta((a, e)) 
		\end{aligned}
		\qquad \qquad 
		\begin{aligned}
			\epsilon(\sigma(e_1, \dots, e_n)) &= \sigma(\epsilon(e_1), \dots, \epsilon(e_n)) \\
			\epsilon(\mu v~e) &= \epsilon(e)[\mu v~e/\!/v]
		\end{aligned}\]
		\caption{Operational semantics of process terms. Here, \(v \in V\), \(a \in A\), and \(e,e_i \in \Exp\).\label{fig: coalgebraic operational semantics}}
	\end{figure}

	Intuitively, the symbol \(0\) is the designated constant of \(S\) denoting the \emph{deadlock} process, which takes no action. 
	Output variables are used in one of two ways, depending on the expression in which they appear.
	A variable \(v\) is \emph{free} in an expression \(e\) if it does not appear within the scope of \(\mu v\) and \emph{bound} otherwise. 
	If \(v\) is free in \(e\), then \(v\) denotes ``output \(v\)''. 
	Otherwise, \(v\) denotes a \emph{goto} statement that returns the computation to the \(\mu v\) that binds \(v\).  
	The process \(\sigma(e_1,\dots, e_n)\) is the process that branches into \(e_1, \dots, e_n\) using an \(n\)-ary operation \(\sigma\) as the branching constructor. 
	The expression \(a e\) denotes the process that performs the action \(a\) and then proceeds with \(e\).
	Finally, \(\mu v~e\) denotes recursion in the variable \(v\).
	
	\paragraph*{Small-step Semantics}
	Next we give a small-step (operational) semantics to process terms that is uniformly defined for the process types in our parametrised family.
	Many of the algebraic theories we consider lack a familiar presentation, which ultimately prevents the corresponding semantics from taking the traditional form of a set of inference rules describing transition relations. 
	We take an abstract approach instead by defining a \(B_M\)-coalgebra structure \(\epsilon : \Exp \to B_M\Exp\) that mirrors the intuitive descriptions of the executions of process terms above.
	The formal description of \(\epsilon\) is summarised in \cref{fig: coalgebraic operational semantics}.
	
	The operational interpretation of the recursion operators requires further explanation.
	Intuitively, \(\mu v~e\) performs the process denoted by \(e\) until it reaches an exit in channel \(v\), at which point it loops back to the beginning. 
	However, this is really only an accurate description of recursion in \(v\) when \(e\) performs an action before exiting in \(v\). 
	For example, the process \(\mu v~v\) not only never exits in channel \(v\), but it also never performs any action at all.
	Thus, the operational interpretation of \(\mu v~v\) is indistinguishable from that of deadlock.
	We deal with this issue as follows: if an exit in channel \(v\) is immediately reached by a branch of \(e\), then we replace that exit with deadlock in \(\mu v~e\).
	Formally, we say that a variable \(v\) is \emph{guarded} in a process term \(e\) if (i) \(e \in V\setminus\{v\}\), (ii) \(e = af\) or (iii) \(e = \mu v~f\) for some \(f \in \Exp\), or (iv) either \(e = \mu u~e_1\) or (v) \(e = \sigma(e_1, \dots, e_n)\) and \(v\) is guarded in \(e_i\) for each \(i \le n\).
	In our calculus, we syntactically allow for recursion in unguarded variables, but one should keep in mind that those variables are ultimately deadlock under the recursion operator. 
	
	The operational interpretation of recursion is formally defined using a \emph{guarded syntactic substitution operator} \([g/\!/v] : B_M\Exp \to B_M\Exp\),\footnote{Technically, it is only partially defined. See \cref{sec: substitution operators} for details.} a variant of the usual syntactic substitution of variables.  
	Given \(g \in \Exp\), we first define \([g/\!/v]\) by induction on \(S^*(V + A\times\Exp)\) as
	\[\begin{aligned}
		u[g/\!/v] &= \begin{cases}
			\eta(u) &u\neq v \\
			\eta(0) &u = v
		\end{cases}
	\end{aligned}
	\qquad\qquad
	\begin{aligned}
		(a, f)[g/\!/v] &= (a, f[g/v]) \\
		\sigma(p_1,\dots,p_n)[g/\!/v] &= \sigma(p_1[g/\!/v], \dots, p_n[g/\!/v])		
	\end{aligned}
	\]
	where \(u \in V\), \(p_i \in S^*(V+A\times\Exp)\), \(f \in \Exp\), and \([g/v]\) replaces free occurrence of \(v\) with \(g\). 
	The following lemma completes the description of the operational semantics of process terms. 
	
	\begin{restatable}{lemma}{guardedsyntacticsubstitution}\label{lem:guardedsubstitution}
		For any \(g \in \Exp\) and \(v \in V\), the map \([g/\!/v]\) factors uniquely through \(B_M\Exp\).
	\end{restatable}

	For more information on these substitution operators, see \cref{sec: substitution operators}.
		
	Formally, the map \(\epsilon\) assigns to each process term \(e\) an \(\Eq\)-congruence class \(\epsilon(e)\) of terms from \(S^*(V + A\times \Exp)\).
	A term from \(S^*(V + A\times \Exp)\) is a combination of variables \(v\) and transition-like pairs \((a, e_i)\), so there is often only a small conceptual leap from the coalgebra structure \(\epsilon\) to a more traditional representation of transitions as decorated arrows.
	We provide the following examples as illustrations of this phenomenon, as well as the specification languages and operational semantics of terms defined above.\footnote{See \cref{sec: examples}.}
	 
	\begin{restatable}{example}{exampleACF}\label{eg:ACF}
		The \emph{algebra of control flows}, or \acro{ACF}, is obtained from the theory of guarded semilattices of \cref{eg:guarded semilattices} and \(M = (1+Id)^{\At}\). 
		Given a structure map \(\beta : X \to B_{(1+\Id)^\At}X\) and \(b \subseteq \At\), write \(x \tr{b\mid a} y\) if \(\beta(x)(\xi) = (a,y)\) for all \(\xi \in b\), and \(x \out{b} v\) if \(\beta(x)(\xi) = v\) for all \(\xi \in b\).
		The operational semantics returns the constant map \(\lambda \xi.v\) given a variable \(v \in V\) and interprets conditional choice as guarded union. 
		For example, let \(e = \mu w~ (a_1(v+_b a_2w)+_b u)\) and \(f = v+_b a_2e\).
		The process denoted by \(e\) is
		\[\begin{tikzpicture}
			\node (0) {\(u\)};
			\node[state, right=2em of 0] (1) {\(e\)};
			\node[state, right=2cm of 1] (2) {\(f\)};
			\node[right=2em of 2] (3) {\(v\)};
			\draw (1) edge[->,bend left=10] node[above] {{\(b\mid a_1\)}} (2);
			\draw (2) edge[->,bend left=10] node[below] {{\(\bar b\mid a_2\)}} (1);
			\draw (1) edge[-implies, double, double distance=0.5mm] node[above] {{\(\bar b\)}} (0);
			\draw (2) edge[-implies, double, double distance=0.5mm] node[above] {{\(b\)}} (3);
		\end{tikzpicture}\]
	\end{restatable}

	\begin{restatable}{example}{exampleAPA}\label{eg:apa}
		The \emph{algebra of probabilistic actions}, or \acro{APA}, is obtained from the theory of pointed convex algebras of \cref{eg:pointed convex algebras} and $M=\D(1+Id)$.
		For a structure map \(\beta : X \to B_{\D(1+\Id)}X\), write \(x \tr{k\mid a} y\) when \(\beta(x)(a,y) = k\) and \(e \out k v\) when \(\beta(e)(v) = k\). 
		The operational semantics returns the Dirac distribution \(\delta_v\) for \(v \in V\) and interprets probabilistic choice as the Minkowski sum.
		The process denoted by \(e = \mu v~ (a_1u +_{\frac{1}{2}}(a_2v +_{\frac{1}{3}} w))\) is
		
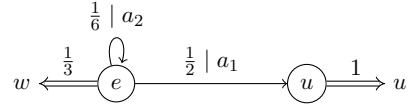
\[\begin{tikzpicture}
			\node (0) {\(w\)};
			\node[state, right=2em of 0] (1) {$e$};
			\node[state, right=2cm of 1] (2) {\(u\)};
			\node[right=2em of 2] (3) {\(u\)};
			
			\draw (1) edge[->] node[above,pos=0.5] {{\( \frac{1}{2}\mid a_1\)}} (2);
			\draw (1) edge[-implies, double, double distance=0.5mm] node[above] {{\(\frac{1}{3}\)}} (0);
			\draw[loop above] (1) edge[->] node[above] {{\( \frac{1}{6}\mid a_2\)}} (1);
			\draw (2) edge[-implies, double, double distance=0.5mm] node[above] {{\(1\)}} (3);
		\end{tikzpicture}\]
	\end{restatable}

	\begin{restatable}{example}{exampleANPA}
		The \emph{algebra of nondeterministic probabilistic actions}, or \acro{ANP}, is obtained from the theory of pointed convex semilattices of \cref{eg:pointed convex semilattices}.
		For a structure map \(\beta : X \to B_{\C}X\), write \(x \mathrel{\raisebox{-0.1em}{\({\to}\circ\)}{\trd{k\mid a}}} y\) to mean there is a \(\theta \in \beta(x)\) such that \(\theta(a,y)=k\), and \(x \out k v\) to mean there is a \(\theta \in \beta(x)\) with \(\theta(v) = k\). 
		The operational semantics returns \(\eta^\C(v)\) given \(v \in V\), interprets nondeterministic choice as convex union, and replaces probabilistic choice with Minkowski sum. 
		For example, \(e = \mu v~((a_1v +_{\frac{1}{3}} a_2w) + a_2v )\) denotes
		
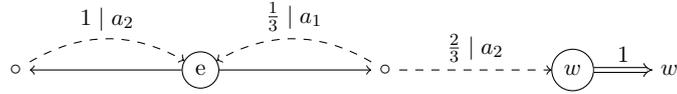
\[\begin{tikzpicture}
			\node (0) {\(\circ\)};
			\node (1) [state, right=2cm of 0] {e};
			\node (2) [right=2cm of 1] {\(\circ\)};
			\node (3) [state, right=2cm of 2] {\(w\)};
			\node (4) [right=2em of 3] {\(w\)};
			
			\draw (1) edge[->] node[above, pos=0.5] {} (0);
			\draw (1) edge[->] node[above, pos=0.5] {} (2);
			\draw (0) edge[->, dashed, bend left] node[above, pos=0.5] {\(1\mid a_2\)} (1);
			\draw (2) edge[->, dashed, bend right] node[above, pos=0.5] {\(\frac{1}{3}\mid a_1\)} (1);
			\draw (2) edge[->, dashed] node[above, pos=0.5] {\(\frac{2}{3}\mid a_2\)} (3);
			\draw (3) edge[-implies, double, double distance=0.5mm] node[above] {\(1\)} (4);
		\end{tikzpicture}\]
	\end{restatable}

\section{Behavioural Equivalence and the Final Coalgebra} \label{sec:final coalgebras}
	In this section, we relate the operational semantics arising from the coalgebra
	structure on $\Exp$ in the previous section to a denotational semantics,
	which arises through the definition of a suitable algebra structure on the 
	domain of process behaviours.
	
	For an arbitrary functor \(B : \Sets \to \Sets\), a \textit{behaviour} is a state of the \textit{final} \(B\)-coalgebra \((Z,\zeta)\), the unique (up to isomorphism) coalgebra (if it exists) such that there is exactly one homomorphism \(!_\beta : (X,\beta) \to (Z,\zeta)\) from every \(B\)-coalgebra \((X,\beta)\).
	It follows from general considerations that the functor \(B_M\) admits a final coalgebra~\cite{rutten2000universal}.
	The universal property of the final \(B_M\)-coalgebra produces the homomorphism \(\beh_\epsilon : (\Exp, \epsilon) \to (Z, \zeta)\).
	The behaviour \(\beh_\epsilon(e)\) is called the \emph{final (coalgebra) semantics} of \(e\), also known as its operational semantics~\cite{ruttenT1993finalsemantics}.

	For example, the final \(B_{\P}\)-coalgebra consists of bisimulation equivalence classes of finite and infinite labelled trees of a certain form~\cite{barr1993terminal}.
	In this setting, \((\Exp, \epsilon)\) is a labelled transition system and the final semantics \(\beh_\epsilon\) constructs a tree from a process term by unrolling.
	Intuitively, this captures the behaviour of a specification by encoding all possible actions and outgoing messages at each time-step in its execution.
	
	\begin{figure}[t!]
		\[
		\begin{aligned}
			\zeta(\gamma(v)) &= [v]_{\Eq} 
			\\
			\zeta(\gamma(a,t)) &= [(a,t)]_{\Eq}
		\end{aligned}
		\qquad\qquad
		\begin{aligned}
			\zeta(\gamma(\sigma(t_1, \dots, t_n))) &= [\sigma(\zeta(t_1), \dots, \zeta(t_n))]_{\Eq}
			\\
			\zeta(\gamma(\mu v~t)) &= \zeta(t)\{\gamma (\mu v~t)/\!/v\}
		\end{aligned}\]
		\caption{The \(\Sigma_M\)-algebra structure of \((Z, \gamma)\). 
			Here, \(v \in V\), \(a \in A\), \(t,t_i \in Z\) for \(i \le n\), and \(\sigma\) is an \(n\)-ary operation from \(S\).
			By Lambek's lemma~\cite{lambek1968fixpoint}, \(\zeta : Z \to B_MZ\) is a bijection, so the first three equations determine \(\gamma : V + SZ + A\times V \to Z\).
			The fourth is a behavioural differential equation~\cite{rutten1998coinduction}.
			\label{fig: algebraic semantics}}
	\end{figure}
	
	In addition to forming the state space of the final \(B_M\)-coalgebra, the set of process behaviours also carries the structure of a \(\Sigma_M\)-algebra \((Z,\gamma)\), summarised in \cref{fig: algebraic semantics}.
	Now, \((\Exp, \alpha)\) is the \textit{initial} \(\Sigma_M\)-algebra, which in particular means there is a unique algebra homomorphism \(\sem- : (\Exp,\alpha) \to (Z,\gamma)\).
	The behaviour \(\sem e\) is called the \textit{initial (algebra) semantics} of \(e\)~\cite{goguenTWW1977initial},
	and provides a denotational semantics to our process calculus.
	
	The algebra structure \(\gamma : \Sigma_M Z \to Z\) of \((Z,\gamma)\) can be seen as a reinterpretation of the programming constructs of the language \(\Exp\) that mimics the operational semantics of process terms.
	The basic constructs are the content of the first three equations in \cref{fig: algebraic semantics}: output variables are evaluated so as to behave like the variables of \((\Exp, \epsilon)\), the behaviour \(at\) performs \(a\) and moves on to \(t\), and \(\sigma(t_1, \dots, t_n)\) branches into the behaviours \(t_1, \dots, t_n\) with additional structure determined by the operation \(\sigma\).
	Interpreting recursive behaviours like \(\mu v~t\) requires coalgebraic analogues of syntactic and guarded syntactic substitution from \cref{sec:specifications of processes}.
	
	For a given behaviour \(s \in Z\) and a variable \(v \in V\), the \emph{behavioural substitution} of \(s\) for \(v\) is the map \(\{s/v\} : Z \to Z\) defined
	%\footnote{This is an example of a coinductive definition in the sense of \cite{rutten1998coinduction}.} 
	by the identity
	\[
		\zeta(t\{s/v\}) = \begin{cases}
			\zeta(s) &\zeta(t) = [v]_{\Eq} \\
			[u]_{\Eq} &\zeta(t) = [u]_{\Eq} \neq [v]_{\Eq} \\
			[(a, r\{s/v\})]_{\Eq} &\zeta(t) = [(a,r)]_{\Eq} \\
			\sigma(\zeta(t_1\{s/v\}), \dots, \zeta(t_n\{s/v\})) &\zeta(t) = [\sigma(\zeta(t_1), \dots, \zeta(t_n))]_{\Eq}
		\end{cases}
	\]
	for any \(t \in Z\).
	%This substitution operator is a generalisation of the one that appears in \cite{milner1984complete}.
	The \emph{guarded behavioural substitution} of \(s\) for \(v\) is constructed in analogy with guarded syntactic substitution from the previous section.
	We start by defining guarded behavioural substitution in \(S^*(V+A\times Z)\) as
	\begin{align*}
		u\{s/\!/v\} &= \begin{cases}
		 	u &u\neq v \\
		 	0 &u = v
		\end{cases}
		&
		\begin{aligned}
			(a,r)\{s/\!/v\} &= (a, r\{s/v\}) \\
			\sigma(r_1, \dots, r_n)\{s/\!/v\} &= \sigma(r_1\{s/\!/v\}, \dots, r_n\{s/\!/v\})
		\end{aligned}
	\end{align*}
	where \(u \in V\), \(a \in A\), and \(r,r_i \in Z\) for \(i \le n\).
	This map lifts to an operator \(B_MZ \to B_MZ\) for the same reason as the guarded syntactic substitution operator. This completes the description of the algebraic structure of \((Z,\gamma)\) in \cref{fig: algebraic semantics}.
	
	\begin{restatable}{theorem}{compositionalitytheorem}\label{thm:compositionality}
		Let \(\sem-\) be the unique algebra homorphism \((\Exp, \alpha) \to (Z,\gamma)\).
		For any process term \(e \in \Exp\), we have \(\beh_\epsilon(e) = \sem e\).
	\end{restatable}

	In other words, the final semantics given with respect to operational rules in \(\Exp\) coincides with the initial semantics given with respect to the programming constructs in \(Z\). 
	Consequently, we write \(\sem-\) in place of \(\beh_\epsilon\) and simply refer to \(\sem e\) as the semantics of \(e\).

	\section{An Axiomatisation of Behavioural Equivalence}\label{sec:soundness and completeness}
	An important corollary of \cref{thm:compositionality} is that behavioural equivalence is a \(\Sigma_M\)-congruence on \((\Exp,\alpha)\), meaning that it is preserved by all the program constructs of \(\Sigma_M\).
	This opens the door to the possibility of deriving behavioural equivalences between process terms from just a few axioms.
	The purpose of this section is to show that all behavioural equivalences between process terms can be derived from the equations in \(\Eq\) presenting \((M,\eta,\mu)\) as well as three axiom schemas concerning the recursion operators. 
		
	The first two out of the three recursion axiom schemas are 
	\[\begin{aligned}
		\text{(\acro R1)} && \mu v~e &= e[\mu v~e/\!/v]
	\end{aligned}
	\qquad\qquad
	\begin{aligned}
	\prftree[l]{(\acro R2)\quad}{\text{\(w\) not free in \(e\)}}{\mu v~e = \mu w~(e[w/v])}
	\end{aligned}
	\]
	Above, \(e[\mu v~e/\!/v]\) is the expression obtained by replacing every guarded free occurrence of \(v\) in \(e\) with the expression \(\mu v~e\) and every unguarded occurrence of \(v\) in \(e\) with \(0\), in analogy with the operator on \(B_M\Exp\) of the same name.\footnote{Indeed, the identity \(\epsilon(e[\mu v~e/\!/v]) = \epsilon(e)[\mu v~e/\!/v]\) holds for all \(e \in \Exp\) and \(v \in V\)~\cref{lem:guardingtheunguarded}.}
	
	The axiom (\acro R1) essentially allows for a sort of guarded unravelling of recursive terms.
	This has the effect of identifying \(\mu v~v\) with \(0\), for example, as well as \(\mu v~av\) with \(a(\mu v~av)\). 
	The latter satisfies our intuition that \(\mu v~av\) should solve the recursive specification \(x = ax\) in the indeterminate \(x\). 
	The axiom (\acro R2) allows for recursion variables to be swapped for fresh variables. 
	This amounts to the observation that pairs of terms like \(\mu v~av\) and \(\mu w~aw\) should both denote the unique solution to \(x = ax\). 
		
	The third recursion axiom schema can be stated in the form of the proof rule
	\[\begin{aligned}
		\prftree[l]{(\acro R3)\quad}{g = e[g/v] \quad \text{\(v\) guarded in \(e\)}}{g = \mu v~e}
	\end{aligned}\]
	We let \(\mathsf R\) denote the set of equations derived from (\acro R1)-(\acro R3), and we let \(\equiv\) denote the smallest congruence in \((\Exp,\alpha)\) containing the set of equations derived from \(\Eq\) and \(\mathsf R\).
	When we refer to examples like \acro{ARB}, \acro{ACF}, \acro{APA}, and \acro{ANP}, we are often identifying each of these with their associated algebras \((\Expm, \hat\alpha)\) of process terms modulo \(\equiv\). 
	
	\paragraph*{Soundness} We would like to argue that \(\equiv\) is \emph{sound} with respect to behavioural equivalence, meaning that \(\sem e = \sem f\) whenever \(e \equiv f\). 
	This is indeed the case, and can be derived from the fact that the set of congruence classes of process terms itself forms a \(B_M\)-coalgebra.
	For an arbitrary function \(h : X \to Y\), call the set \(\ker(h) := \{(x,y) \mid h(x) = h(y)\}\) the \emph{kernel} of \(h\).
	
	\begin{restatable}{lemma}{soundnesstheorem}\label{thm:soundness}
		The congruence \(\equiv\) is the kernel of a coalgebra homomorphism.
	\end{restatable}
	
	We write \([-]_\equiv : \Exp \to \Expm\) for the quotient map and \((\Expm, \bar \epsilon)\) for the coalgebra structure on \(\Expm\) making \([-]_{\equiv}\) a coalgebra homomorphism (there is at most one such coalgebra structure~\cite{rutten2000universal}).
	As \(\sem-\) is the unique coalgebra homomorphism \((\Exp,\epsilon) \to (Z, \zeta)\), and because there is also a coalgebra homomorphism \(\beh_{\bar \epsilon} : (\Expm, \bar\epsilon) \to (Z,\zeta)\), it must be the case that \(\beh_{\bar \epsilon} \circ [-]_\equiv = \sem-\).
	By \cref{thm:soundness}, if \(e \equiv f\), then 
	\(
	\sem e = \beh_{\bar \epsilon}([e]_\equiv) = \beh_{\bar \epsilon}([f]_\equiv) = \sem f
	\).
	This establishes the following.
	
	\begin{theorem}[Soundness]\label{thm:soundness optics}
		Let \(e,f \in \Exp\).
		If \(e \equiv f\), then \(\sem e = \sem f\).
	\end{theorem}

	Soundness allows us to derive at least a subset of all the behavioural equivalences between process terms from the axioms in \(\Eq\) and \(\mathsf R\).
	If our aspiration were simply to have a set of behaviour-preserving code-transformations, then we could simply stop here and be satisfied, since in principle we could see the axioms of \(\Eq\) and \(\mathsf R\) as rewrite rules that satisfy this purpose.

	\paragraph*{Completeness} 
	Aiming a bit higher than deriving only a subset of the behavioural equivalences between process terms, we move on to show the converse of \cref{thm:soundness optics}, that \(\equiv\) is \emph{complete} with respect to behavioural equivalence.
	We use \cite[Lemma 5.1]{schmid2021star}, which can be stated as follows.
	
	\begin{lemma}\label{lem:global approach}
		Let \(B : \Sets \to \Sets\) be an endofunctor with a final coalgebra \((Z,\zeta)\), and let \(\mathbf C\) be a class of \(B\)-coalgebras. 
		If \(\mathbf C\) is closed under homomorphic images\footnote{Ie., if \((X,\beta) \in \mathbf C\) and \(h : (X, \beta) \to (Y, \vartheta)\), then \((h[X],\vartheta|_{h[X]}) \in \mathbf C\).} and has a final object \((E,\varepsilon)\), then \(\beh_\varepsilon : E \to Z\) is injective. 
	\end{lemma}
	
	A \emph{subcoalgebra} of a \(B\)-coalgebra \((X, \beta)\) is an injective map \(\iota : U \hookrightarrow X\) such that \(\beta|_{U}\) factors through \(B(\iota)\). 
	A \(B\)-coalgebra is \emph{locally finite} if every of its states is contained in (the image of) a finite subcoalgebra.
	We instantiate \cref{lem:global approach} in the case where \(B = B_M\), \((E,\varepsilon) = (\Expm, \bar\epsilon)\), and \(\mathbf C\) is the class of locally finite \(B_M\)-coalgebras.
	Completeness of \(\equiv\) with respect to behavioural equivalence follows shortly after, for if \(\sem e = \sem f\), then \(\beh_{\bar\epsilon}([e]_\equiv) = \beh_{\bar \epsilon}([f]_\equiv)\).
	By \cref{lem:global approach}, \(\beh_{\bar \epsilon}\) is injective, so \([e]_\equiv = [f]_\equiv\) or equivalently \(e \equiv f\). 
	To establish the converse of \cref{thm:soundness optics}, it suffices to show that our choices of \((E,\varepsilon)\) and \(\mathbf C\) satisfy the hypotheses of \cref{lem:global approach}.
	
	Before we continue, we would like to remind the reader of Assumption \ref{asm:finite arities}, that \(S\) only has operations of finite arity, as up until now it has not been strictly necessary.

	\begin{restatable}{lemma}{localfinitenesscondition}\label{lem:local finiteness}
		The coalgebra \((\Exp, \epsilon)\) is locally finite. 
	\end{restatable}

	The class of locally finite coalgebras is closed under homomorphic images: if \((X,\beta)\) is locally finite and \(h : (X,\beta) \to (Y,\vartheta)\) is a surjective homomorphism, then for any \(y \in Y\) and \(x \in X\) such that \(h(x) = y\), and for any finite subcoalgebra \(U\) of \((X,\beta)\) containing \(x\), \(h[U]\) is a finite subcoalgebra of \((Y,\vartheta)\) containing \(y\)~\cite{gumm1999elements}.
	Since \(y\) was arbitrary, it follows from \cref{thm:soundness} that \((\Expm, \bar \epsilon)\) is locally finite.
	 
	What remains to be seen among the hypotheses of \cref{lem:global approach} is that \((\Expm, \bar \epsilon)\) is the \emph{final} locally finite coalgebra, meaning that for any locally finite coalgebra \((X,\beta)\) there is a unique coalgebra homomorphism \((X,\beta) \to (\Expm,\bar\epsilon)\). 
	Every homomorphism from a locally finite coalgebra is the union of its restrictions to finite subcoalgebras, so it suffices to see that every finite subcoalgebra of \((X,\beta)\) admits a unique coalgebra homomorphism into \((\Expm, \bar \epsilon)\).
	
	To this end, we make use of an old idea, possibly originating in the work of Salomaa~\cite{salomaa1966two}.
	We associate with every finite coalgebra a certain system of equations whose  solutions (in \(\Expm\)) are in one-to-one correspondence with coalgebra homomorphisms into \((\Expm,\bar \epsilon)\).
	Essentially, if a system admits a unique solution, then its corresponding coalgebra admits a unique homomorphism into \((\Expm, \bar{\epsilon})\).
	This would then establish finality.
	 
	\begin{definition}
		A \emph{(finite) system of equations} is a sequence of the form \(\{x_i = e_i\}_{i \le n}\) where \(x_i \in V\) and \(e_i \in \Exp\) for \(i \le n\), and none of \(x_1,\dots, x_n\) appear as bound variables in any of \(e_1,\dots,e_n\). 
		A system of equations \(\{x_i=e_i\}_{i \le n}\) is \emph{guarded} if \(x_1,\dots,x_n\) are guarded in \(e_i\) for each \(i \le n\).
		A \emph{solution} to \(\{x_i=e_i\}_{i \le n}\) is a function \(\phi : \{x_1,\dots,x_n\} \to \Exp\) such that \[
		\phi(x_i) \equiv e_i[\phi(x_1)/x_1, \dots, \phi(x_n)/x_n]
		\] 
		for all \(i \le n\) and \(x_1,\dots, x_n\) do not appear free in \(\phi(x_i)\) for any \(i \le n\).  
	\end{definition}
	
	Every finite \(B_M\)-coalgebra \((X,\beta)\) gives rise to a guarded system of equations in the following way:
	for each \(p \in S^*(V + A\times X)\), define \(p^\dagger\) inductively as
	\[
		v^\dagger = v
		\qquad
		(a,e)^\dagger = ae
		\qquad 
		\sigma(f_1,\dots,f_n)^\dagger = \sigma(f_1^\dagger, \dots, f_n^\dagger)
	\]
	and for each \(x \in X\), let \(p_x\) be a representative of \(\beta(x)\).  
	The\footnote{Technically speaking, there could be many systems of equations associated with a given coalgebra. We say ``the'' system of equations because any two have the same set of solutions up to \(\Eq\).} system of equations \emph{associated with} \((X,\beta)\) is then defined to be \(\{x = p_x^\dagger\}_{x \in X}\). 
	We treat the elements of \(X\) as variables in these equations, and note that by definition every \(y \in X\) is guarded in \(p_x^\dagger\).
	
	\begin{restatable}{theorem}{solutionsarehomomorphisms}\label{lem:solutions are homoms}
		Let \((X,\beta)\) be a finite \(B_M\)-coalgebra and \(\phi : X \to \Exp\) a function.
		Then the composition \([-]_\equiv\circ\phi : X \to \Expm\) is a \(B_M\)-coalgebra homomorphism if and only if \(\phi\) is a solution to the system of equations associated with \((X,\beta)\).
	\end{restatable}

	As a direct consequence of \cref{lem:solutions are homoms}, we see that a finite subcoalgebra \(U \hookrightarrow \Exp\) of \((\Exp,\epsilon)\) is a solution to the system of equations associated with \((U, \epsilon|_U)\). 
	
	\begin{example}
		The system of equations associated with the automaton in \cref{eg:ACF} is the two-element set \(\{x_1 = a_1x_2 +_b u, x_2 = v +_b a_2x_1 \}\).
		The map \(\phi : \{x_1, x_2\} \to \Exp\) defined by \(\phi(x_1) = \mu w~(a_1(v +_b a_2w)+_b u)\) and \(\phi(x_2) = v +_b a_2~\phi(x_1)\) is a solution. 
	\end{example}

	\cref{lem:solutions are homoms} establishes a one-to-one correspondence between solutions to systems and coalgebra homomorphisms as follows.
	Say that two solutions \(\phi\) and \(\psi\) to a system \(\{x_i=e_i\}_{i \le n}\) are \emph{\(\equiv\)-equivalent} if \(\phi(x_i) \equiv \psi(x_i)\) for all \(i \le n\).
	Starting with a solution \(\phi : X \to \Exp\) to the system associated with \((X,\beta)\), we obtain the homomorphism \([-]_\equiv\circ \phi\) using \cref{lem:solutions are homoms}. 
	A pair of solutions \(\phi\) and \(\psi\) are \(\equiv\)-equivalent if and only if \([-]_\equiv\circ\phi = [-]_\equiv\circ \psi\), so up to \(\equiv\)-equivalence the correspondence \(\phi \mapsto [-]_\equiv\circ \phi\) is injective. 
	Going in the opposite direction and starting with a homomorphism \(\psi : (X,\beta) \to (\Expm, \bar{\epsilon})\), let \(e_x\) be a representative of \(\psi(x)\) for each \(x \in X\) and define \(\phi := \lambda x.e_x\).
	Then \(\phi\) is a solution to \((X,\beta)\), and \([-]_\equiv \circ \phi = \psi\).
	Thus, up to \(\equiv\)-equivalence, solutions to systems are in one-to-one correspondence with coalgebra homomorphisms into \((\Expm,\bar{\epsilon})\).
	
	Say that a system \emph{admits a unique solution up to \(\equiv\)} if it has a solution and any two solutions to the system are \(\equiv\)-equivalent. 
	Since, up to \(\equiv\)-equivalence, solutions to a system associated with a coalgebra \((X,\beta)\) are in one-to-one correspondence with coalgebra homomorphisms \((X,\beta) \to (\Expm,\bar{\epsilon})\), it suffices for the purposes of satisfying the hypotheses of \cref{lem:global approach} to show that every finite guarded system of equations admits a unique solution up to \(\equiv\). 
	The following theorem is a generalisation of~\cite[Theorem 5.7]{milner1984complete}. 
	
	\begin{restatable}{theorem}{milnerslemma}\label{lem:Milner's Lemma}
		Every finite guarded system of equations admits a unique solution up to \(\equiv\).
	\end{restatable}

	The proof is a recreation of the one that appears under \cite[Theorem 5.7]{milner1984complete} with the more general context of our paper in mind.
	Remarkably, the essential details of the proof remain unchanged despite the jump in the level of abstraction between the two results. 
	
	Completeness of \(\equiv\) with respect to behavioural equivalence is now a direct consequence of \cref{lem:global approach,lem:solutions are homoms,lem:Milner's Lemma}.

	\begin{corollary}[Completeness]
		Let \(e,f \in \Exp\).
		If \(\sem e = \sem f\), then \(e \equiv f\). 
	\end{corollary}

	One way to interpret this theorem is that the algebra \((\Expm, \hat\alpha)\) of process terms modulo \(\equiv\) is isomorphic to a subalgebra of \((Z,\gamma)\), or dually \((\Expm,\bar\epsilon)\) is a subcoalgebra of \((Z,\zeta)\).
	It is in this sense that \acro{ARB}, \acro{ACF}, \acro{APA}, and \acro{ANP} are algebras of behaviours. 

	\section{Star Fragments}\label{sec:star fragments}
	In this section we study a fragment of our specification languages consisting of \emph{star expressions}. 
	These include primitive actions from \(A\), a form of sequential composition, and analogues of the Kleene star. 
	We do not aim to give a complete axiomatisation of behavioural equivalence for star expressions, as even in simple cases this is notoriously difficult.
	Nevertheless, we think it is valuable to extrapolate from known examples a speculative axiomatisation independent of the specification languages from previous sections.
	
	Fix an algebraic theory \((S,\Eq)\) and assume \(S\) consists of only constants and binary operations.
	Its \emph{star fragment} is the set \(\SExp\) of expressions given by the grammar
	\[
	e,e_i ::= c \mid \skiptt \mid a \mid e_1 +_\sigma e_2 \mid e_1e_2 \mid e^{(\sigma)}
	\]
	where \(a \in A\), \(c\) is a constant in \(S\), and \(\sigma\) is a binary \(S\)-operation.
	
	The star fragment of an algebraic theory is a fragment of \(\Exp\) in the sense that star expressions can be thought of as shorthands for process terms, as we explain next. In this translation, we fix a distinguished variable \(\unit \in V\), called the \emph{unit}, which will denote successful termination, and we also fix a  variable \(v\) distinct from the unit, which will appear in the fixpoint.  
	The translation of star expressions to process terms is defined to be
	\[
		\skiptt \mapsto \unit
		\qquad 
		a \mapsto a\unit
		\qquad 
		e_1 +_\sigma e_2 \mapsto \sigma(e_1, e_2)
		\qquad
		e_1e_2 \mapsto e_1[e_2/\unit]
		\qquad
		e^{(\sigma)} \mapsto \mu v~(e[v/\unit]+_\sigma \unit)
	\]
	Sequential composition of terms is associative and distributes over branching operations on the right-hand side\footnote{But not on the left-hand side! Observe the difference between the processes \(a(b + c)\) and \(ab + ac\) here.}: for any \(e_1,e_2,f \in \SExp\), \((e_1 +_\sigma e_2)f\) and \(e_1 f +_\sigma e_2f\) translate to the same process term. 
	Similarly, the intuitively correct identities \(\skiptt e = e = e\skiptt\) hold modulo translation, as well as the identity \(0e = 0\).\footnote{But not \(e0 = 0\)! See also the previous footnote.\label{fn:grat}} 
	
	\begin{figure}[!t]
		\centering
		\[\begin{aligned}	
			\ell(c) &= [c]_{\Eq}
			\\
			\ell(\skiptt) &= [\checkmark]_{\Eq}
			\\
			\ell(a) &= [(a,\skiptt)]_{\Eq}  
		\end{aligned}
		\qquad\qquad
		\begin{aligned}
			\ell(e_1 +_\sigma e_2) &= \sigma(\ell(e_1), \ell(e_2))
			\\
			\ell(ef) &= p(\ell(f),[(a_1,e_1f)]_{\Eq}, \dots, [(a_n,e_nf)]_{\Eq})
			\\
			\ell(e^{(\sigma)}) &= p([0]_{\Eq},[(a_1,e_1e^{(\sigma)})]_{\Eq}, \dots, [(a_n,e_ne^{(\sigma)})]_{\Eq}) +_\sigma [\checkmark]_{\Eq}
		\end{aligned}\]
		\caption{
			The coalgebra structure map \(\ell : \SExp \to L_M\SExp\). 
			Here, \(c\) is a constant of \(S\), \(\sigma\) is a binary operation of \(S\), \(a \in A\), and \(e,e_i \in \SExp\). In the last two equations, \(\ell(e) = [p(\checkmark,(a_1,e_1), \dots, (a_n,e_n))]_{\Eq}\) for some \(p \in S^*(\{\checkmark\} + A\times \SExp)\).
		\label{fig:operational semantics of SExp}}
	\end{figure} 

	The operational semantics for star expressions is given by an \(L_M\)-coalgebra \((\SExp, \ell)\) in \cref{fig:operational semantics of SExp}, where \(L_M = M(\{\checkmark\} + A \times \Id)\).
	Abstractly, the operational interpretation \(\ell(e)\) of a star expression \(e\) is obtained by translating \(e\) into a process term (also called \(e\)) and then identifying \(\unit\) with \(\checkmark\) in \(\epsilon(e)\).
	While the notation is somewhat opaque at this level of generality, in specific instances the map \(\ell\) amounts to a familiar transition structure.
	
	\begin{example}
		The star fragment of \acro{ACF} from \cref{eg:guarded semilattices} and \cref{eg:ACF} coincides with \acro{GKAT}, the algebra of programs introduced in \cite{kozentseng2008} and studied further in \cite{gkat,schmidkappekozensilva2021gkat}.
		Instantiating \(\SExp\) in this context reveals the syntax
		\[
			e_i ::= 0 \mid \skiptt \mid a \mid e_1 +_b e_2 \mid e_1 e_2 \mid e^{(b)}
		\]
		for \(b \subseteq \At\) and \(a \in A\). 
		This is nearly the syntax of \acro{GKAT}, the only difference being the presence of \(\skiptt\) and \(0\) instead of Boolean constants \(b \subseteq \At\).
		This is merely cosmetic, as we can just as well define \(b := \skiptt +_b 0\).
		
		In this context, \(M = (1 + \Id)^{\At}\), and so \(L_M \cong (2 + A \times \Id)^{\At}\), which is the precise coalgebraic signature of the automaton models of \acro{GKAT} expressions.
		It is readily checked that the operational semantics of \acro{GKAT} also coincides with the operational semantics of the star fragment of \acro{ACF} given above.
	\end{example}
	
	\begin{example}
		The star fragment of \acro{APA} from \cref{eg:pointed convex algebras} and \cref{eg:apa} is a subset of the calculus of programs introduced in \cite{bonchiSV2019tracesfor}, but with an iteration operator for each \(p \in [0,1]\).
		Instantiating \(\SExp\) in this context reveals the syntax
		\[
		e_i ::= 0 \mid \skiptt \mid a \mid e_1 +_p e_2 \mid e_1 e_2 \mid e^{(p)}
		\]
		for \(p \in [0,1]\) and \(a \in A\). 
		The process \(e^{(p)}\) can be thought of as a generalised Bernoulli process that runs \(e\) until it reaches \(\checkmark\) and then flips a weighted coin to decide whether to start from the beginning of \(e\) or to terminate successfully. 
	\end{example} 
	
	We now provide a candidate axiomatisation for the star fragment while leaving the question of completeness open. 	
	Say that a star expression \(e\) is \emph{guarded} if the unit is guarded in \(e\) as an expression in \(\Exp\).
	We define \(\Eq^*\) to be the theory consisting of \(\Eq\), the axiom schema
	\[
	\begin{aligned}
		(\mathsf{E^*1}) && 1e &= e1 = e  \\
		(\mathsf{E^*2}) && ce &= c 	  
	\end{aligned}
	\qquad\qquad
	\begin{aligned}
		(\mathsf{E^*3}) && e_1 (e_2e_3) &= (e_1e_2)e_3 \\
		(\mathsf{E^*4}) && (e +_\tau \skiptt)^{(\sigma)} &= (e +_\tau 0)^{(\sigma)}
	\end{aligned}
	\]
	and the inference rules
	\[
		\prftree[l]{(\(\mathsf{E^*5}\))\quad}{\text{\(e\) is guarded}}{e^{(\sigma)} = ee^{(\sigma)} +_\sigma \skiptt}
		\qquad\qquad 
		\prftree[l]{(\(\mathsf{E^*6}\))\quad}{g = eg +_\sigma f \qquad \text{\(e\) is guarded}}{g = e^{(\sigma)}f}
	\]
	In the specific cases where \(\Eq = \mathsf{SL}\) and \(\Eq = \mathsf{GS}\), \(\Eq^*\) is equivalent to the candidate axiomatisations for the star fragments of \acro{ARB}~\cite{milner1984complete} and \acro{ACF}~\cite{gkat,schmidkappekozensilva2021gkat}.\footnote{See \cref{sec:unguarded unravelling} for details.}
	
	There is a difference between our axioms and the axioms in \cite{milner1984complete,gkat,schmidkappekozensilva2021gkat}: instead of (\acro{E^*5}), all equations of the form \(e^{(\sigma)} = ee^{(\sigma)} +_\sigma 1\) appear in loc cit, even those where \(e\) is unguarded.
	We adopt (\acro{E^*5}) instead because the unrestricted version of (\acro{E^*5}) fails to be sound for the star expressions of \acro{APA}.
	For example, if \(e = \skiptt +_{\frac 13} a\), then \(ee^{(1/2)} +_{\frac 12} \skiptt \out{7/{12}} \checkmark\) while \(e^{(\frac 12)} \out{1/2} \checkmark\).
	Secondly, the unrestricted axioms can be derived from (\acro{E^*5}) in the cases of Milner's star fragment and the star fragment of \acro{ACF}.\footnote{See \cref{sec:unguarded unravelling}.}
	
	We are confident that a completeness result can be obtained in several instances of the framework for the axiomatisation we have suggested above.
	However, in several cases this cannot happen.
	For example, there is no way to derive the identity \(((a +_{\frac12} \skiptt) + b)^* = ((a +_{\frac12} 0) + b)^*\) from \(\mathsf{CS}^*\) (see \cref{eg:pointed convex semilattices}) despite these expressions being behaviourally equivalent. 
	What is likely missing from \(\mathsf{CS}\) is a number of axioms that would allow \(\skiptt\) to be moved to the top level of every \(S\)-term (and then replaced by \(0\) using (\acro{E^*4})). 
	Algebraic theories where this is doable are called \textit{skew-associative}, which we define formally as follows.
	
	\begin{definition}\label{def:skew-associative}
		 An algebraic theory \((S,\Eq)\) consisting of constants and binary operations is called \emph{skew-associative} if for any pair of binary operations \(\sigma_1,\tau_1\), there is a pair of binary operations \(\sigma_2, \tau_2\) such that \(\sigma_1(x,\tau_1(y, z)) = \tau_2(\sigma_2(x, y), z)\) appears in \(\Eq\).
	\end{definition}
	
	Many of the examples we care about are skew-associative, including the theories of semilattices, guarded semilattices, and convex algebras.
	
	\begin{question}
		Assume \((S,\Eq)\) is a skew-associative algebraic theory.
	 	If \(e\) and \(f\) are behaviourally equivalent star expressions, is it true that \(\mathsf{E^*} \vdash e = f\)? 
	\end{question}

	\section{Related Work}\label{sec:related work}
	Our framework can be seen as a generalisation of Milner's \acro{ARB}~\cite{milner1984complete} that reaches beyond nondeterministic choice and covers several other process algebras already identified in the literature.	
	For example, instantiating our framework in the theory of pointed convex algebras produces the algebra we have called \acro{APA} (see \cref{eg:apa}), which only differs from the algebra \acro{PE} of Stark and Smolka~\cite{stark1999complete} in the axiom (\acro{R1}).
	In loc cit, the requirement that the variable be guarded in the recursed expression is absent because recursion is computed as a least fixed point in their semantics.
	This is not how we interpret recursion.
	We have included the guardedness requirement because it is necessary for the soundness of the axiom in our semantics: for example, where \(e = u +_{\frac12} v\), we have \(\mu v~e \out{1/2} u\) and \(e[\mu v~e/v] \out{3/4} u\).
	In contrast, both \(\mu v~e\) and \(e[\mu v~e/v]\) exit in \(u\) with probability \(1\) in \cite{stark1999complete}.
	
	For another example, instantiating our framework in the theory \acro{CS} of pointed convex semilattices gives \acro{ANP} (see \cref{eg:pointed convex semilattices}), which differs from the calculus of Mislove, Ouaknine, and Worrell~\cite{mislove03axioms} on two points.
	Firstly, their axiomatisation contains an unguarded version of (\acro{R1}), like in \cite{stark1999complete}.
	Secondly, the underlying algebraic theory of \cite{mislove03axioms} corresponds to \acro{CS} extended with the axiom \(x +_p 0 = 0\). 
	The resulting theory is known in the literature as that of \emph{convex semilattices with top}~\cite{bonchiSV2019tracesfor}.
	
	Star expressions for non-deterministic processes appeared in the work of Milner \cite{milner1984complete} as a fragment of \acro{ARB} and can be thought of as a bisimulation-focused analogue of Kleene's regular expressions for NFAs. 
	While the syntaxes of Milner's star expressions and Kleene's regular expressions are the same, there are several important differences between their interpretations. 
	For example, sequential composition is interpreted as the variable substitution \(ef := e[f/\unit]\) in Milner's paper, which fails to distribute over \(+\) on the left. 
	A notable insight from \cite{milner1984complete} is that, despite these differences, an iteration operator \((-)^*\) can be defined for Milner's star expressions that satisfies many of the same identities as the Kleene star. 
	Given a variable \(v\) distinct from the unit and a process term \(e\) of \acro{ARB} in which at most the unit is free,
	\[
		e^* = \mu v~(e[v/\unit] + \unit)
	\]
	defines the iteration operator in Milner's star fragment of \acro{ARB}.
	In \Cref{sec:star fragments}, we generalised this construction of Milner for the more general process types that we considered in this paper. Our proposed axiomatization is also inspired by Milner's work. We expect completeness of our general calculus will be a hard problem, as completeness in the instantiation to \acro{ARB} was open for decades despite the extensive literature on the subject~\cite{fokkinkzantema1994basic,fokkink1997perpetual,fokkinkzantema1997termination,baetenCG2006starheight,grabmayerfokkink2020complete}.

	There are clear parallels between our work and the thesis of Silva~\cite{silva2010kleene}, in which a family of calculi is introduced that includes one-exit versions of \acro{ARB}, \acro{ACF}, and \acro{APA} (see \cref{eg:arb,eg:ACF,eg:apa}). 
	The main difference is that our framework is parametric on a finitary monad on \(\Sets\) whereas Silva's is centered around one particular theory (semilattices). 
	However, her work considers general polynomial functors on \(\Sets\), which we have not yet done in our paper. 
	We could achieve a similar level of generality by replacing \(A\times\Id\) in our signatures \(\Sigma_M\), \(B_M\), and \(L_M\) with an arbitrary polynomial functor.
	
	Our results are also in the same vein as the work of Myers on coalgebraic expressions~\cite{myers2009coalgebraic}.
	Coalgebraic expressions generalise the calculi of \cite{silva2010kleene} to arbitrary finitary coalgebraic signatures on a variety of algebras, and furthermore have totally defined recursion operators similar to ours.
	However, the focus of the framework of coalgebraic expressions is on language semantics, achieved by lifting the coalgebraic signature to a variety. This distinguishes the framework from our approach: we focus on bisimulation semantics.
	This focus is also the reason we interpret our \(B_M\)-coalgebras in \(\Sets\) and not in the Kleisli category
	of the monad $M$, as is done in~\cite{hasuoJS2007generic} to capture trace semantics of coalgebras.
	
	Finally, there is also a notable connection to the iterative theories of Elgot~\cite{elgot1975monadic,bloomE1976iterative,bloomE1988varieties,nelson1983iterative}. 
	\Cref{lem:Milner's Lemma} in particular implies that our process algebras are examples of iterative algebras.

\section{Future Work}\label{sec:conclusion}
	In this paper, we introduced a family of process types whose branching structure is determined by an algebraic theory.
	We provided each process type with a fully expressive specification language paired with a sound and complete axiomatisation of behavioural equivalence.

	There are several instantiations of our framework that we have not yet explored and are of interest.  
	For example, processes with multiset branching given by the theory of commutative monoids produces nondeterministic processes with a simplistic notion of resources. Another example is nondeterministic weighted processes with branching captured by the monad arising from the weak distributive law between the free semimodule and powerset monads~\cite{bonchi2021combining}.
	Yet another instantiation arises from the theory of monoids (presenting the list monad), which produces processes related to breadth-first search algorithms.  

	Star fragments offer a uniform construction of Kleene-like algebras for a variety of paradigms of computing.
	However, our framework does not suggest an axiomatisation of the star fragment that combines nondeterministic and probabilistic choice, as the theory \acro{CS} is not skew-associative (see \cref{def:skew-associative}).
	We would like to expand our framework to include this fragment as it provides an interesting but nonstandard interpretation of a part of the language ProbNetKAT used to verify probabilistic networks~\cite{foster16probabilistic}. 

	We would also like to investigate the question at the end of \cref{sec:star fragments} of whether \acro{E^*} is complete for skew-associative theories. 
	In particular, we believe that a connection can be made to the work of Grabmayer and Fokkink~\cite{grabmayerfokkink2020complete} on LLEE-charts, which provides a completeness theorem for the so-called 1-free expressions of the star fragment of \acro{ARB}.
	Our process algebras also have uniformly defined 1-free star fragments, and it is not difficult to give 1-free versions of the axiomatisation \acro{E^*}.
	We intend to suitably generalise LLEE-charts to arbitrary skew-associative theories and prove completeness theorems for 1-free star fragments.

	Finally, we would like to know whether our operational semantics for process terms is an instance of the mathematical operational semantics introduced by Turi and Plotkin~\cite{turiplotkin1997operational}.

\printbibliography

\appendix

\section{Notes on Guarded Semilattices} \label{sec: guarded semilattices}

This appendix is here mainly to give an account of what is essentially the algebra of \texttt{if-then-else} statements of a propositional imperative programming language.
For a fixed finite set \(\At\), call a structure \((X,0,\{+_b\}_{b \subseteq \At})\) consisting of an underlying set \(X\), a constant \(0 \in X\), and a binary operation \(+_b : X^2 \to X\) for each \(b \subseteq \At\) a \emph{guarded semilattice} if \((X,0,\{+_b\}_{b\subseteq\At})\) satisfies all instances of the equations (\acro{GS}1)-(\acro{GS}4) from \cref{eg:guarded semilattices}.
Guarded semilattices are examples of the so-called \emph{McCarthy algebras} of \cite{manes1991ifthenelse}, and conversely every McCarthy algebra in finitely many propositions is a guarded semilattice.
We change the name to emphasise the presence of finite Boolean guards, as well as the inherent order structure of guarded semilattices that will be expanded upon in future work.

Let \(\mathbf{GS}\) denote the category of guarded semilattices and their homomorphisms.
We define the functor \(F : \Sets \to \mathbf{GS}\) such that \(FX = ((1 + X)^{\At}, 0, \{+_b\}_{b\subseteq \At})\) for every set \(X\), where \(0 := \lambda \xi.0\) and 
\[
	h +_b k := \lambda\xi.\begin{cases}
			h(\xi) &\xi \in b \\
			k(\xi) &\xi \not\in b
		\end{cases}
\]
and given a function \(f : X \to Y\), \(F(f)(h) = (1 + f) \circ h\).
It is straightforward to verify that \(FX\) is indeed a guarded semilattice for any \(X\).

\begin{lemma} \label{lem:gsl embedding}
	Every guarded semilattice can be embedded into a guarded semilattice of the form \(FX\) for some set \(X\).
\end{lemma}

\begin{proof}
	Let \((X,0, \{+_b\}_{b\subseteq \At})\) be a guarded semilattice and define the map \(\eta : X \to (1 + X)^{\At}\) by \(\eta(x) = \lambda \xi.x+_\xi 0\).
	Let \(\At = \{\xi_1, \dots, \xi_n\}\).
	It follows from (\acro{GS}4) and (\acro{GS}2) that for any \(x \in X\),
	\[
		x = x +_{\xi_n} (x +_{\xi_{n-1}} (\dots(x +_{\xi_1} 0)))
	\]
	Whence, \(\eta\) is clearly injective, for if \(x +_\xi 0 = y +_\xi 0\) for all \(\xi \in \At\), one can show by induction on \(n\) that\[
	 x +_{\xi_n} (x +_{\xi_{n-1}} (\dots(x +_{\xi_1} 0))) = y +_{\xi_n} (y +_{\xi_{n-1}} (\dots(y +_{\xi_1} 0)))
	\]
	Hence, it suffices to see that \(\eta\) is an algebra homomorphism.
	Given \(x,y \in X\) and \(b \subseteq \At\), we have
	\begin{align*}
		\eta({x +_b y})(\xi) 
		&= (x +_b y) +_\xi 0 
		= x +_{b \cap \xi} (y +_\xi 0)
		= \begin{cases}
			x +_\xi (y +_\xi 0) &\xi \in b \\
			x +_0 (y +_\xi 0) &\xi \not\in b
		\end{cases}\\
		&= \begin{cases}
			x +_\xi \bar \xi(y +_\xi 0) &\xi \in b \\
			y +_\xi 0 &\xi \not\in b
		\end{cases}\quad
		= \begin{cases}
			x +_\xi 0 &\xi \in b \\
			y +_\xi 0 &\xi \not\in b
		\end{cases}
		\quad = (\eta({x}) +_b \eta({y}))(\xi)
	\end{align*}
\end{proof}

As a corollary, we obtain the following theorem, which essentially states that the theory of guarded semilattices presents the monad \((1 + \Id)^{\At}\).

\begin{theorem}\label{thm:free guarded semi}
	Where \(U : \mathbf{GSL} \to \Sets\) is the forgetful functor taking an algebra to its underlying set, there exists an adjunction \(F \dashv U\).
\end{theorem}

The unit of the adjunction \(\Id \Rightarrow UF\) is the map \(\eta\) defined in the proof of \cref{lem:gsl embedding}, and where \(\At = \{\xi_1, \dots, \xi_n\}\), the counit \(\varepsilon : FU \Rightarrow \Id\) is given by 
\[
	\varepsilon(h) \mapsto h(\xi_n) +_{\xi_n} (h(\xi_{n-1}) +_{\xi_{n-1}} (\dots (h(\xi_1) +_{\xi_1}0)))
\]
Furthermore, the transformation \(\Delta^* : UFUF \Rightarrow UF\), as it is defined in \cref{eg:guarded semilattices}, is precisely the transformation \(U(\varepsilon_{F})\).  

\section{Examples from Section 3}\label{sec: examples}
\begin{itemize}
	\item Consider \(M=(1 + \Id)^{\At}\) and an expression \(e = \mu w~ (a_1(v+_b a_2w)+_b u)\). The \(B_M\)-coalgebra structure on this expression is given by:
	\begin{align*}
		\epsilon(e) &= \epsilon(a_1(v+_b a_2w)+_b u)[e/\!/w] = (\lambda \xi.(a_1, v+_b a_2w) +_b \lambda \xi.u)[e/\!/w] \\
		&= (\lambda \xi.(a_1, v+_b a_2w))[e/\!/v] +_b (\lambda \xi.u)[e/\!/w] = \lambda \xi.(a_1, v+_b a_2e) +_b \lambda \xi.u \\
		&= \lambda \xi.(a_1, f) +_b \lambda \xi.u
\end{align*}
	\item Consider \(M=\D(1+\Id)\) and an expression \(e = \mu v~ (a_1u +_{\frac{1}{2}}(a_2v +_{\frac{1}{3}} w))\). The derivation of the coalgebra structure for this expression is given by:
	\begin{align*}
		\epsilon(e) &= \epsilon(a_1u +_{\frac{1}{2}}(a_2v +_{\frac{1}{3}} w))[e/\!/v] = \frac{1}{2}\epsilon(a_1u)[e/\!/v] + \frac{1}{2}\epsilon(a_2v +_{\frac{1}{3}} w)[e/\!/v] \\
		&= \frac{1}{2}\epsilon(a_1u)[e/\!/v] + \frac{1}{2}(\frac{1}{3}\epsilon(a_2v)[e/\!/v] + \frac{2}{3}\epsilon(w)[e/\!/v]) %= \frac{1}{2}\epsilon(a_1u)[e/\!/v] + \frac{1}{6}\epsilon(a_2v)[e/\!/v] + \frac{2}{6}\epsilon(w)[e/\!/v] 
		\\
		&= \frac{1}{2}\delta_{(a_{1}, u)}[e/\!/v] + \frac{1}{6}\delta_{(a_{2},v)}[e/\!/v] + \frac{2}{6}\delta_{w}[e/\!/v] \\
		&= \frac{1}{2}\delta_{(a_{1}, u)}+ \frac{1}{6}\delta_{(a_{2},e)} + \frac{2}{6}\delta_{w} \\
	\end{align*}
	\item Finally, consider \(M=\C_0\) and an expression \(e = \mu v~((a_1v +_{\frac{1}{3}} a_2w) + a_2v )\). The \(B_M\)-coalgebra structure is given by:
	\begin{align*}
		\epsilon(e) &= \epsilon((a_1v +_{\frac{1}{3}} a_2w) + a_2v ))[e/\!/v] = \conv(\epsilon(a_1v +_{\frac{1}{3}} a_2w)[e/\!/v], \epsilon(a_2v)[e/\!/v])\\
		&= \conv(\epsilon(a_1v)[e/\!/v] +_{\frac{1}{3}} \epsilon(a_2w)[e/\!/v], \epsilon(a_2v)[e/\!/v]) = \conv(\epsilon(a_1e) +_{\frac{1}{3}} \epsilon(a_2w), \epsilon(a_2e))\\
		&= \conv(\{\frac{1}{3}\delta_{(a_1,e)} + \frac{2}{3}\delta_{(a_2,w)}\}, \{\delta_{(a_2,e)}\})= \conv(\{\delta_{(a_1,e)}\} +_{\frac{1}{3}} \{\delta_{(a_2,w)}\}, \{\delta_{(a_1,e)}\})\\
	\end{align*}
\end{itemize}

\section{Notes on Substitution Operators}\label{sec: substitution operators}

In this paper, several different kinds of substitution operators are used. 
We have left this appendix to explain their properties and prepare for their use in proof later in the document.

The first kind of substitution that appears is \emph{syntactic} substitution. 
Given two expressions \(e\) and \(f\) and a variable \(v\), we define the expression \(e[f/v]\) by induction on \(e\) as follows: For the basic constructions, 
\[
	u[f/v] = \begin{cases}
			f &u = v \\
			u &u \neq v
		\end{cases}
	\quad 
	(ae)[f/v] = a(e[f/v])
	\quad
	\sigma(e_1,\dots,e_n)[f/v] = \sigma(e_1[f/v], \dots, e_n[f/v])
\]
but for the recursion case, we only let \((\mu u~e)[f/v]\) be well-defined if either \(u = v\), in which case \((\mu v~e)[f/v] = \mu v~e\) (because \(v\) is not free in \(\mu v~e\)), or \(u\) is not free in \(f\), in which case \((\mu u~e)[f/v] = \mu u~(e[f/v])\).
Thus, \([f/v]\) is a partial map \(\Exp \rightharpoonup \Exp\).

We similarly define \(e[f_1/v_1,\dots,f_n/v_n]\) for a distinct list of variables \(v_1,\dots,v_n\) to be the simultaneous substitution of \(f_i\) for \(v_i\), \(i=1,\dots,n\).
For this kind of substitution,
\[
	u[f_1/v_1,\dots,f_n/v_n] = \begin{cases}
			f_i &u = v_i \\
			u &(\forall i\le n)~u \neq v_i
		\end{cases}
\]
and if \(u = v_i\), then
\[
	(\mu u~e)[f_1/v_1,\dots,f_n/v_n] = (\mu u~e)[f_1/v_1, \dots, f_{i-1}/v_{i-1}, f_{i+1}/v_{i+1},\dots, f_n/v_n]
\]
and otherwise, if \(u\) is not free in \(f_i\) for all \(i\le n\), then
\[
	(\mu u~e)[f_1/v_1,\dots,f_n/v_n] = \mu u~(e[f_1/v_1,\dots,f_n/v_n])
\]
Again, \([f_1/v_1,\dots,f_n/v_n]\) defines a partial operation \(\Exp \rightharpoonup \Exp\). 

\begin{lemma}
	Let \(e, f \in \Exp\) and \(v \in V\).
	If no free variable of \(f\) is bound in \(e\), then \(e[f/v]\) is well-defined. 
\end{lemma}
\begin{proof}
	Variables bound in \(e\) are formally given by a function \(\bv : \Exp \to \mathcal{P}(V)\):
	\[
	\begin{aligned}
		\bv(v) &= \emptyset \\
		\bv(ae) &= \bv(e)\\
	\end{aligned}
	\qquad\qquad
	\begin{aligned}
		\bv(\sigma(e_1, \dots, e_n))&=\bv(e_1)~\cup \dots \cup~\bv(e_n)\\
		\bv(\mu v~e)&=\{v\}~\cup~\bv(e)
	\end{aligned}
	\]
	Similarly, the free variables of expression \(e\) are a function  \(\fv : \Exp \to \mathcal{P}(V)\) defined as:
	\[
	\begin{aligned}
		\fv(v) &= \{v\}\\
		\fv(ae) &= \fv(e)\\
	\end{aligned}
	\qquad\qquad
	\begin{aligned}
		\fv(\sigma(e_1, \dots, e_n))&=\fv(e_1)~\cup \dots \cup~\fv(e_n)\\
		\fv(\mu v~e)&=\fv(e) \setminus \{v\}\\
	\end{aligned}
	\]
	We prove the lemma by induction on \(e\). We assume that \(\fv(f) \cap \bv(e) = \emptyset\).
	\begin{itemize}
		\item The variable case is trivial, as the syntactic substitution is always well-defined.
		\item For the prefixing case, assume that \(e=ag\). By definition \((ag)[f/v]=a(g[f/v])\), so either is well-defined whenever \(g[f/v]\) is well-defined. Since \(\bv(ag)=\bv(g)\), by the induction hypothesis \(ag[f/v]\) is well-defined.
		\item Now suppose \(e=\sigma(e_1, \dots, e_n)\). By definition, \(\sigma(e_1, \dots, e_n)[f/v]=\sigma(e_1[f/v], \dots, e_n[f/v])\), so substituting \(f\) for \(v\) is well defined for each \(e_1, \dots, e_n\). 
		Since \(\bv(e_k) \subseteq \bv(\sigma(e_1, \dots, e_n))\) for each subexpression \(e_k\), for any \(x \in \fv(f)\) we have that \(x \notin \bv(\sigma(e_1, \dots,e_k, 
		\dots, e_n))\) and therefore \(x \notin \bv(e_k)\). By the induction hypothesis, \(e_k[f/v]\) is well-defined for each \(e_k\). It follows that \(e[f/v]\) is well-defined.
		\item For the recursion case, assume that \(e=\mu u~g\). We consider two subcases.
		\begin{itemize}
			\item If \(u = v\), then syntactic substitution is well-defined and is given by \((\mu u~g)[f/v]=\mu u~g\).
			\item Otherwise, \(u \neq v\). By assumption we know that if \(x \in \fv(f)\), then \(x \notin \bv(g)\cup\{u\}\), so \(u\) is not free in \(f\). 
			It follows that \((\mu u~g)[f/v]\) is well defined if and only if \(g[f/v]\) is well-defined. By the induction hypothesis, for any \(x \in \fv(f)\) we have that \(x \notin \bv(g)\). \qedhere
		\end{itemize}
	\end{itemize}
\end{proof}

Semantic substitution has a cousin that appears in the paper, namely \emph{guarded syntactic substitution}. 
Given \(e,f \in \Exp\) and \(v \in V\), we define \(e[f/\!/v]\) to be the expression
\[
u[f/\!/v] = \begin{cases}
	0 &u = v \\
	u &u \neq v
\end{cases}
\qquad 
(ae)[f/\!/v] = a(e[f/v])
\qquad
\sigma(e_1,\dots,e_n)[f/\!/v] = \sigma(e_1[f/\!/v], \dots, e_n[f/\!/v])
\]
Again, we only let \((\mu u~e)[f/\!/v]\) be well-defined if either \(u = v\), in which case \((\mu v~e)[f/\!/v] = \mu v~e\), or \(u\) is not free in \(f\), in which case \((\mu u~e)[f/\!/v] = \mu u~(e[f/\!/v])\).
Thus, \([f/\!/v]\) is yet another partial map \(\Exp \rightharpoonup \Exp\).

The first appearance of the guarded substitution operator in the paper is actually as a partial operator on \(B_M(\Exp)\), however. 
Recalling that \(B_M(\Exp) = M(V + A\times \Exp)\) is the set \(S^*(V + A \times \Exp)\) modulo \(\Eq\), it is defined first as a map \(V + A\times \Exp \to B_M\Exp\) as follows: given \(u \in V\) and \((a,e) \in A\times \Exp\), let
\[
u[f/\!/v] = \begin{cases}
	[0]_{\Eq} &u = v \\
	[u]_{\Eq} &u \neq v
\end{cases}
\qquad 
(a,e)[f/\!/v] = a(e[f/v])
\]
We record the existence of a lift in the following lemma. 

\guardedsyntacticsubstitution*

\begin{proof}
	By definition, \([\mu v~e/\!/v]\) is obtained from the unique lifting of the partial map \(h : V + A\times\Exp \rightharpoonup  S^*(V + A\times \Exp)\) defined 
	\[
	h(u) = \begin{cases}
		u &u \neq v \\
		0 &u = v
	\end{cases}
	\qquad
	h(a,f) = (a, f[\mu v~e/v]) 
	\]
	to a partial \(S\)-algebra homomorphism \(h^\# : S^*(V + A\times \Exp) \rightharpoonup S^*(V + A\times \Exp)\) by further composing with the quotient homomorphism \([-]_{\Eq} : S^*(V + A\times \Exp) \to B_M\Exp\) (this factorisation exists for partial maps because \(S^*\) preserves monos).
	Thus, \[
	\ker([\mu v~e/\!/v]) = \ker([-]_{\Eq} ~ h^\#) \supseteq \ker([-]_{\Eq}) \cap \operatorname{dom}(h^\#)^2
	\]
	where for an arbitrary partial map \(f : X \rightharpoonup Y\), \(\ker(f) = \{(x_1,x_2) \mid f(x_1)= f(x_2)\}\).
	In other words, \([\mu v~e/\!/v]\) is constant on \(\Eq\)-congruence classes.
	Since \(B_M\) preserves monos and \([-]_{\Eq}\) is surjective, there is a unique \(B_M\Exp \rightharpoonup B_M\Exp\) such that the following diagram commutes.
	\[\begin{tikzpicture}
		\node (1) {\(S^*(V + A\times \Exp)\)};
		\node[right=2cm of 1] (2) {\(S^*(V + A\times \Exp)\)};
		\node[below=1cm of 1] (3) {\(B_M\Exp\)};
		\node[below=1cm of 2] (4) {\(B_M\Exp\)};
		\node[above right=0.17cm of 3] (5) {\(\circlearrowleft\)};
		\draw (1) edge[-left to] node[above] {\(h^\#\)} (2);
		\draw (1) edge[->] node[left] {\([-]_{\Eq}\)} (3);
		\draw (1) edge[->] node[fill=white] {\([\mu v~e/\!/v]\)} (4);
		\draw (2) edge[->] node[right] {\([-]_{\Eq}\)} (4);
		\draw (3) edge[-left to, dashed] (4);
	\end{tikzpicture}\qedhere\]
\end{proof}

The two versions of guarded syntactic substitution interact as expected.

\begin{lemma}
	For any \(p \in S^*(V+A\times \Exp)\) and \(f \in \Exp\) and \(v \in V\), \(p[f/\!/v]\) is well-defined if and only if \(p^{\dagger}[f/\!/v]\) is well-defined, and in such a case \([p]_\Eq[f/\!/v]=\epsilon(p^{\dagger}[f/\!/v])\). 
\end{lemma}
\begin{proof}
	We proceed by induction on \(p\).
	\begin{itemize}
		\item Suppose that \(p = w\) for some \(w \in V\). In this case, \(w^{\dagger}=w\), so both \(w[f/\!/v]\) and \(w^{\dagger}[f/\!/v]\) are trivially well-defined. For the desired identity, consider following subcases.
		\begin{itemize}
			\item If \(w = v\), then \([w]_{\Eq}[f/\!/v]=\eta^{M}(0)=\epsilon(0)=\epsilon(w[f/\!/v])=\epsilon(w^\dagger[f/\!/v])\).
			\item Otherwise, if \(w \neq v\), then \([w]_{\Eq}[f/\!/v]=\eta^{M}(w)=\epsilon(w)=\epsilon(w[f/\!/v])=\epsilon(w^\dagger[f/\!/v])\).
		\end{itemize}
		\item Suppose that \(p = (a,e)\) for some \(a \in A\) and \(e \in \Exp\). Since \((a,e)^{\dagger}=ae\) in this case,  both \(p[f/\!/v]\) and \(p^{\dagger}[f/\!/v]\) are well-defined when \(e[f/v]\) is. 
		Therefore \(ae[f/\!/v]\) is well-defined if and only if \((a,e)^{\dagger}[f/\!/v]\) is well defined.
		Furthermore, we know \([(a,e)]_{\Eq}[f/\!/v]=[(a,e[f/v])]_{\Eq}\), so 
		\(
		[(a,e)]_\Eq[f/\!/v] = [(a,e[f/v])]_{\Eq} = \epsilon(ae[f/v]) = \epsilon((a,e)^\dagger[f/\!/v])
		\).
		\item Finally, suppose \(p=\sigma(p_1, \dots, p_n)\) for some \(p_1, \dots, p_n \in S^*(V+A\times \Exp)\), and recall that \(\sigma(p_1, \dots, p_n)^{\dagger}=\sigma(p_1^{\dagger}, \dots, p_n^{\dagger})\). 
		Since \(\sigma(p_1, \dots, p_n)[f/\!/v]=\sigma(p_1[f/\!/v], \dots, p_n[f/\!/v])\), we have 
		\[\sigma(p_1, \dots, p_n)^{\dagger}[f/\!/v]=\sigma(p_1^{\dagger}[f/\!/v], \dots, p_n^{\dagger}[f/\!/v])\] 
		By the induction hypothesis, \(p_i[f/\!/v]\) is well-defined if and only if \(p_i^{\dagger}[f/\!/v]\) is well-defined, for \(1\leq i \leq n\).
		Towards the desired identity, recall that \[
			[\sigma(p_1, \dots, p_n)]_{\Eq}[f/\!/v]=\sigma([p_1]_{\Eq}, \dots, [p_n]_{\Eq})[f/\!/v]=\sigma([p_1]_{\Eq}[f/\!/v], \dots, [p_n]_{\Eq}[f/\!/v])
		\]
		It follows from the induction hypothesis that 
		\begin{align*}
			[\sigma(p_1, \dots, p_n)]_{\Eq}[f/\!/v]
			&= [\sigma(p_1[f/\!/v], \dots, p_n[f/\!/v])]_{\Eq} \\
			&= \sigma(\epsilon(p_1^\dagger[f/\!/v])),\dots, \epsilon(p_n^\dagger[f/\!/v])) \\
			%=\epsilon(\sigma(e_1[f/\!/v],\dots, e_n[f/\!/v]))
			&=\epsilon(\sigma(p_1,\dots, p_n)^\dagger[f/\!/v])
			\qedhere
		\end{align*}
	\end{itemize}
\end{proof}

In the paper, we also introduced two versions of substitution for behaviours. 
Given \(t,s \in Z\), and \(v \in V\), we define \(t\{s/v\}\) by a behavioural differential equation that depends on \(\zeta(t)\) as follows:
\[
\zeta(t\{s/v\}) = \begin{cases}
	\zeta(s) &\zeta(t) = [v]_{\Eq} \\
	[u]_{\Eq} &\zeta(t) = [u]_{\Eq} \neq [v]_{\Eq} \\
	[(a, r\{s/v\})]_{\Eq} &\zeta(t) = [(a,r)]_{\Eq} \\
	\sigma(\zeta(t_1\{s/v\}), \dots, \zeta(t_n\{s/v\})) &\zeta(t) = [\sigma(\zeta(t_1), \dots, \zeta(t_n))]_{\Eq}
\end{cases}
\]
Note that unlike its syntactic relative, \(\{s/v\}\) is a total function \(Z \to Z\). 
Despite their differences, however, behavioural substitution enjoys many of the important properties of syntactic substitution. 
The following theorem provides a simplified coinductive principle with which we can prove these properties for our processes.

\begin{theorem}\label{thm:bisimilar maps}
	Let \(h,k : Z \to Z\).
	If \(h\) and \(k\) satisfy properties (i)-(iii) below, then \(h = k\).
	\begin{enumerate}[leftmargin=1cm]
		\item[(i)] if \(\zeta(t) = [w]_{\Eq}\), then \(\zeta(h(t)) = \zeta(k(t))\);
		\item[(ii)] if \(\zeta(t) = [(a, r)]_{\Eq}\), then \(\zeta(h(t)) = [(a,h(r))]_{\Eq}\) and \(\zeta(k(t)) = [(a,k(r))]_{\Eq}\); and
		\item[(iii)] if \(\zeta(t) = \sigma(\zeta(t_1), \dots, \zeta(t_n))\), then 
		\[
			\zeta(h(t)) = \sigma(\zeta(h(t_1)), \dots, \zeta(h(t_n)))
		\quad
		\text{and}
		\quad
			\zeta(k(t)) = \sigma(\zeta(k(t_1)), \dots, \zeta(k(t_n)))
		\] 
	\end{enumerate}
\end{theorem}

\begin{proof}
	We proceed with a proof by \emph{coinduction}. 
	Namely, we will give the relation\[
	R := \Delta_Z \cup \{(h(r), k(r)) \mid r \in Z\}
	\]
	a \(B_M\)-coalgebra structure \(\rho : R \to B_MR\) such that the projections \(\pi_i : R \to Z\), \(i = 1,2\), are coalgebra homomorphisms. 
	By finality of \((Z,\zeta)\), this then implies that \(\pi_1 = {!}_{\rho} = \pi_2\), which establishes the identity we are hoping to prove.
	
	Consider a \(t \in Z\) and let \(\zeta(t) = [p(v_1,\dots,v_n, (a_1,s_1),\dots,(a_m,s_m))]_{\Eq}\) for some \(p \in S^*(V + A\times Z)\).
	Along the diagonal, define
	\[
	\rho(t,t) = [p(v_1,\dots,v_n, (a_1, (s_1,s_1)), \dots, (a_n, (s_n,s_n)))]_{\Eq}
	\]
	For the other pairs, assume \([v_i]_{\Eq} = \zeta(h(t_i)) = \zeta(k(t_i))\) for \(i = 1,\dots, n\) using (i) and write
	\[
	\rho(h(t), k(t)) = [p(v_1, \dots, v_n, (a_1, (h(s_1), k(s_1))), \dots, (a_m, (h(s_m), k(s_m))))]_{\Eq}
	\]   
	using (ii) and (iii).
	To see that \(\pi_i\) is a coalgebra homomorphism for each \(i = 1,2\), observe
	\begin{align*}
	B_M(\pi_1)(\rho(h(t),k(t)))
	&= [S^*(\pi_1)(p(v_1, \dots, v_n, (a_1, (h(s_1), k(s_1))), \dots, (a_m, (h(s_m), k(s_m)))))]_{\Eq} \\
	&= [p(v_1,\dots,v_n, (a_1, h(s_1), \dots, (a_m, h(s_m))]_{\Eq} \\
	&= \zeta \circ \pi_1(h(r), k(r))
	\end{align*}
	and similarly for \(\pi_2\).
\end{proof} 

For the following lemma, we say that a variable \(v\) is \emph{dead} in a behaviour \(t\) if for any other behaviour \(s\), \(t\{s/v\} = t\). 

\begin{lemma}\label{lem:composing behavioural substitutions}
	Let \(u,v \in V\) and \(r,s,t \in Z\).
	If \(v\) is dead in \(t\) and \(u \neq v\), then \[r\{s/v\}\{t/u\} = r\{t/u\}\{s\{t/u\}/v\}\] 
\end{lemma}

\begin{proof}
	We use \cref{thm:bisimilar maps}, which requires us to verify (i)-(iii) for the maps \(\{s/v\}\{t/u\}\) and \(\{t/u\}\{s\{t/u\}/v\}\).
	Assume \(u, v, w\) are distinct variables, and let \(a \in A\).
	There are several cases to consider. 
	\begin{itemize}[leftmargin=1cm]
		\item[(i)] If \(\zeta(r) = [v]_{\Eq}\), then \(\zeta(r\{s/v\}) = \zeta(s)\) and \(\zeta(r\{t/u\}) = \zeta(r)\). 
		This means that
		\[
		\zeta(r\{s/v\}\{t/u\}) = \zeta(s\{t/u\}) = \zeta(r\{s\{t/u\}/v\}) = \zeta(r\{t/u\}\{s\{t/u\}/v\})
		\]
		\item[(i')] If \(\zeta(r) = [u]_{\Eq}\), then \(\zeta(r\{s/v\}) = \zeta(r)\) and \(\zeta(r\{t/u\}) = \zeta(t)\).
		This means that 
		\[
			\zeta(r\{s/v\}\{t/u\}) = \zeta(r\{t/u\}) = \zeta(t) = \zeta(t\{s\{t/u\}/v\}) = \zeta(r\{t/u\}\{s\{t/u\}/v\})
		\]
		\item[(i'')] If \(\zeta(r) = [w]_{\Eq}\), then
		\(
			\zeta(r\{s/v\}\{t/u\}) = \zeta(r) = \zeta(r\{t/u\}\{s\{t/u\}/v\})
		\).
		\item[(ii)] If \(\zeta(r) = [(a, r')]_{\Eq}\), then \(\zeta(r\{s/v\}) = [(a,r'\{s/v\})]_{\Eq}\) and \(\zeta(r\{t/u\}) = [(a, r'\{t/u\})]_{\Eq}\).
		It follows that \[
		\zeta(r\{s/v\}\{t/u\}) = [(a, r'\{s/v\}\{t/u\})]_{\Eq}
		\] 
		and \[
		\zeta(r\{t/u\}\{s\{t/u\}/v\}) = [(a, r'\{t/u\}\{s\{t/u\}/v\})]_{\Eq}
		\]
		\item[(iii)] Now let
		\(
			\zeta(r) = \sigma(\zeta(r_1), \dots, \zeta(r_n))
		\).
		By definition,
		\[
			\zeta(r\{s/v\}\{t/u\}) = \sigma(\zeta(r_1\{s/v\}\{t/u\}), \dots, \zeta(r_n\{s/v\}\{t/u\}))
		\]
		and
		\[
			\zeta(r\{t/u\}\{s\{t/u/v\}) = \sigma(\zeta(r_1\{t/u\}\{s\{t/u/v\}), \dots, \zeta(r_n\{t/u\}\{s\{t/u/v\}))
		\]
		as desired.\qedhere
	\end{itemize}
\end{proof}

\begin{lemma}
	For any \(r,s,t \in Z\) and \(v \in V\), \(r\{s/v\}\{t/v\} = r\{s\{t/v\}/v\}\).
\end{lemma}

\begin{proof}
	This also follows from \cref{thm:bisimilar maps}, where this time we take \(h = \{s/v\}\{t/v\}\) and \(k = \{s\{t/v\}/v\}\).
\end{proof}

We also found the need to define a \emph{guarded} version of behavioural substitution. 
Given \(u \in V\) and \((a,r) \in A\times Z\), we define
\begin{align*}
	u\{s/\!/v\} &= \begin{cases}
		u &u\neq v \\
		0 &u = v
	\end{cases} &
	(a,r)\{s/\!/v\} &= (a, r\{s/v\}) 
\end{align*}
and denote by \(\{s/\!/v\}\) the unique lifting of this map to an operator on \(B_MZ\).\footnote{By Lambek's lemma, \(\zeta:Z \cong B_MZ\), so \(\{s/\!/v\}\) may also denote the corresponding operator \(\zeta^{-1}~\{s/\!/v\}~\zeta\) on \(Z\) (but we don't really find occasion for this here).}

\begin{lemma}\label{lem:gdd sub of dead variables}
	Let \(v \in V\) and \(t,s \in Z\). 
	If \(v\) is dead in \(t\), then \(\zeta(t)\{s/\!/v\} = \zeta(t)\). 	
\end{lemma}

\begin{proof}
	Let \(\zeta(t)=[p]_{\Eq}\) for some \(p \in S^*(V + A\times Z)\). 
	We proceed by induction on \(p\).
	\begin{itemize}
		\item For the variable case, suppose \(p = w \neq v\). 
		We have \([w]_{\Eq}\{s/\!/v\}=[w]_{\Eq}=\zeta(t)\).
		\item Now assume that \(p=(a,e)\) for some \(e \in \Exp\). 
		Since \(v\) is dead in \(t\), \([(a,e)]_{\Eq}\{s/\!/v\}=[(a, e\{s/v\})]_{\Eq} = \zeta(t\{s/v\}) = \zeta(t)\).   
		\item For the inductive step, assume that \(p=\sigma(\zeta(t_1), \dots, \zeta(t_n))\). 
		By the induction hypothesis, \begin{align*}
			\zeta(t)\{s/\!/v\} 
			&= [\sigma(\zeta(t_1)\{s/\!/v\}, \dots, \zeta(t_n)\{s/\!/v\})]_{\Eq} \\
			&= [\sigma(\zeta(t_1), \dots, \zeta(t_n))]_{\Eq} 
			=\zeta(t) \qedhere
		\end{align*}
	\end{itemize}
\end{proof}

\begin{lemma}\label{lem:semantic fp distributes}
	Let \(u,v\) be distinct variables and \(s,t \in Z\). 
	If \(v\) is dead in \(s\), then \[(\mu v~t)\{s/u\} = \mu v~(t\{s/u\})\] 
\end{lemma}

\begin{proof}
	The behaviour \(\mu v~(t\{s/u\})\) is (by definition) the unique solution to the behavioural differential equation \[
		\zeta(r) = \zeta(t\{s/u\})\{r/\!/v\} \tag{*}
	\]
	in the variable \(r\). 
	Thus, it suffices to see that \(r = (\mu v~t)\{s/u\}\) satisfies this equation. 
	To this end, there are a few cases to consider. 
	\begin{itemize}
		\item If \(\zeta(t) = [v]_{\Eq}\), then \(
			\zeta(\mu v~t) = \zeta(t)\{\mu v~t/\!/v\} = [0]_{\Eq}
		\).
		Hence, \[
			\zeta((\mu v~t)\{s/u\}) 
			= [0]_{\Eq} 
			= \zeta(t)\{(\mu v~t)\{s/u\}/\!/v\}
			= \zeta(t\{s/u\})\{(\mu v~t)\{s/u\}/\!/v\}
		\]
		since \(u\) is dead in \(t\) as well.
		Setting \(r = (\mu v~t)\{s/u\}\) satisfies (*), as desired. 
		\item If \(\zeta(t) = [u]_{\Eq}\), then \(\zeta(\mu v~t) = \zeta(t)\{\mu v~t/\!/v\} = \zeta(t)\). 
		Taking \(r = (\mu v~t)\{s/u\}\), we see that 
		\[
			\zeta(r) = \zeta(t\{s/u\}) = \zeta(s) = \zeta(s)\{r/\!/v\} = \zeta(t\{s/u\})\{r/\!/v\}
		\]
		by \cref{lem:gdd sub of dead variables}, since \(v\) is dead in \(s\). 
		\item Let \(w\) be distinct from \(u\) and \(v\).
		If \(\zeta(t) = [w]_{\Eq}\), then trivially \(\zeta(r) = [w]_{\Eq} = \zeta(t) = \zeta(t\{s/u\}) = \zeta(t\{s/u\}) \{r/\!/v\}\).
		\item If \(\zeta(t) = [(a, t')]_{\Eq}\), then \(\zeta(\mu v~t) = [(a, t'\{\mu v~t/v\})]_{\Eq}\).
		Where \(r := (\mu v~t)\{s/u\}\),
		\begin{align*}
			\zeta((\mu v~t)\{s/u\}) &= [(a, t'\{(\mu v~t)/v\}\{s/u\}\})]_{\Eq} \\
			&= [(a, t'\{s/u\}\{(\mu v~t)\{s/u\}/v\}\})]_{\Eq} \tag{\cref{lem:composing behavioural substitutions}}\\
			&= [(a, t'\{s/u\})]_{\Eq}\{r/\!/v\} = \zeta(t\{s/u\})\{r/\!/v\}
		\end{align*}
		because \(v\) is dead in \(s\). 
		\item Finally, let \(\zeta(t) = [p(v_1, \dots, v_n, (a_1, r_1), \dots, (a_m,r_m))]_{\Eq}\).
		We have
		\begin{align*}
			\zeta(r) 
			&= \zeta((\mu v~t)\{s/u\}) \\
			&= [p(w_1, \dots, w_n, (a_1, r_1\{\mu v~t/v\}\{s/u\}), \dots, (a_m,r_m\{\mu v~t/v\}\{s/u\}))]_{\Eq} \\
			&= [p(w_1, \dots, w_n, (a_1, r_1\{s/u\}\{r/v\}), \dots, (a_m,r_m\{s/u\}\{r/v\}))]_{\Eq} \\
			&= [p(v_1', \dots, v_n', (a_1, r_1\{s/u\}), \dots, (a_m,r_m\{s/u\}))]_{\Eq}\{r/\!/v\} = \zeta(t\{s/u\})\{r/\!/v\}
		\end{align*}
		where \(v_i' = v_i\) if \(v_i \neq u\) else \(v_i' = u\), and \(w_i = v_i'\) if \(v_i' \neq v\) else \(w_i=0\). \qedhere
	\end{itemize}
\end{proof}

\begin{lemma}\label{lem:freeness and liveness}
	Let \(v \in V\) and \(e \in \Exp\). 
	If \(v\) is not free in \(e\), then \(v\) is dead in \(\sem e\).
\end{lemma}

\begin{proof}
	We show that \(\zeta(\sem e\{s/v\}) = \zeta(\sem e)\) for all \(s\) by induction on \(e\). 
	\begin{itemize}
		\item In the variable case, we only consider \(u \neq v\). 
		Here, \(\zeta(\sem u) = [u]_{\Eq}\), so \(\zeta(\sem{u\{s/v\}}) = \sem u\).
		\item Now suppose the result is true for \(e\). 
		Since \(v\) is free in \(ae\) if and only if \(v\) is free in \(e\), it must be the case that \(v\) is not free in \(e\). 
		By the induction hypothesis, \(\zeta(\sem{ae}\{s/v\}) = [(a, \sem e\{s/v\})]_{\Eq} = [(a, \sem e)]_{\Eq}\).
		\item Next, assume the result for \(e_1, \dots, e_n\), and let \(\sigma\) be an \(S\)-operation. 
		Since \(v\) is not free in \(\sigma(e_1,\dots,e_n)\) if and only if \(v\) is not free in any of the \(e_i\), and \(\zeta(\sem{\sigma(e_1, \dots, e_n)}) = \sigma(\zeta(\sem{e_1}), \dots, \zeta(\sem{e_n})) \) we have 
		\begin{align*}
			\zeta(\sem{\sigma(e_1, \dots, e_n)}\{s/v\}) &= \sigma(\zeta(\sem{e_1}\{s/v\}), \dots, \zeta(\sem{e_n}\{s/v\})) \\
			&= \sigma(\zeta(\sem{e_1}), \dots, \zeta(\sem{e_n})) = \zeta(\sem{\sigma(e_1, \dots, e_n)}) 
		\end{align*}
		
		\item Now assume the result for \(e\) and let \(u \in V\). 	
		In this case, we consider the expression \(\mu u~e\) and the following two subcases. 
		\begin{itemize}
			\item If \(u = v\), then by assumption \(v\) is not free in \(\mu u~e\).
			Here, \(\zeta(\sem{\mu v~e}) = \zeta(e)\{\sem {\mu v~e}/\!/v\}\).
			Let \(p(v_1, \dots, v_n, (a_1, t_1), \dots, (a_m,t_m)) \in S^*(V + A \times Z)\) such that 
			\[\zeta(\sem e) = [p(v_1, \dots, v_n, (a_1, t_1), \dots, (a_m,t_m))]_{\Eq}\] 
			and set \(w_i = 0 \) if \(v_i = v\) else \(w_i = v_i\).
			By definition, 
			\begin{align*}
				\zeta(\sem {\mu v~e}) &= \zeta(e)\{\sem {\mu v~e}/\!/v\} \\
				&= [p(w_1, \dots, w_n, (a_1, t_1\{\sem {\mu v~e}/v\}), \dots, (a_m,t_m\{\sem {\mu v~e}/v\}))]_{\Eq}
			\end{align*}
			because \(w_i \neq v\) for any \(i \le n\), and consequently,
			\begin{align*}
				&\zeta(\sem {\mu v~e}\{s/v\}) \\
				&= [p(w_1, \dots, w_n, (a_1, t_1\{\sem {\mu v~e}/v\}\{s/v\}), \dots, (a_m,t_m\{\sem {\mu v~e}/v\}\{s/v\}))]_{\Eq} \\
				&= [p(w_1, \dots, w_n, (a_1, t_1\{\sem {\mu v~e}\{s/v\}/v\}), \dots, (a_m,t_m\{\sem {\mu v~e}\{s/v\}/v\}))]_{\Eq} \\
				&= [p(w_1, \dots, w_n, (a_1, t_1), \dots, (a_m,t_m))]_{\Eq} \{\sem {\mu v~e}\{s/v\}/\!/v\} \\
				&= \zeta(\sem e)\{\sem {\mu v~e}\{s/v\}/\!/v\}
			\end{align*}
			Now \(\sem{\mu v~e}\) is the unique solution to the behavioural differential equation \(\zeta(r) = \zeta(\sem e)\{r/\!/v\}\) in the indeterminate \(r\), and \(r = \sem{\mu v~e}\{s/t\}\) satisfies this equation, so it must be the case that \(\sem {\mu v~e} = \sem{\mu v~e}\{s/v\}\). 
			Hence, \(v\) is dead in \(\sem {\mu v~e}\). 
			
			\item Now assume \(u \neq v\). 
			This means that \(v\) is free in \(\mu u~e\) if and only if it is free in \(e\), so by the inductive hypothesis \(v\) is dead in \(\sem e\). 
			Again, we let \[
				\zeta(\sem e) = [p(v_1, \dots, v_n, (a_1, t_1), \dots, (a_m,t_m))]_{\Eq}
				\] 
			and \(w_i = 0 \) if \(v_i = v\) else \(w_i = v_i\).
			In the previous case, we showed that \(u\) is dead in \(\sem{\mu u~e}\), so we begin by showing that \(\sem{\mu u~e}\{s/v\} = \sem{\mu u~e}\) for each \(s \in Z\) such that \(u\) is dead in \(s\). 
			We have
			\begin{align*}
				&\zeta(\sem{\mu u~e}\{s/v\})\\
				&= [p(w_1, \dots, w_n, (a_1, t_1\{\sem {\mu u~e}/u\}\{s/v\}), \dots, (a_m,t_m\{\sem {\mu u~e}/u\}\{s/v\}))]_{\Eq} \\
				&= [p(w_1, \dots, w_n, (a_1, t_1\{s/v\}\{\sem {\mu u~e}\{s/v\}/u\}), \dots,\\
				&\hspace{4cm} (a_m,t_m\{s/v\}\{\sem {\mu u~e}\{s/v\}/u\}))]_{\Eq} \tag{\cref{lem:composing behavioural substitutions}}\\
				&= [p(v_1, \dots, v_n, (a_1, t_1\{s/v\}), \dots, (a_m,t_m\{s/v\}))]_{\Eq} \{\sem {\mu u~e}\{s/v\}/\!/u\} \\
				&= \zeta(\sem e) \{\sem {\mu u~e}\{s/v\}/\!/u\}
				\tag{\(v\) dead in \(\sem{e}\)}
			\end{align*}
			Hence, \(r = \sem {\mu u~e}\{s/v\}\) solves the behavioural differential equation defining \(\sem{\mu u~e}\), so \(\sem{\mu u~e} = \sem {\mu u~e}\{s/v\}\).
			The desired result follows from the following observation: Both \(u\) and \(v\) are clearly dead in \(\sem 0\), so for arbitrary \(s \in Z\) we have 
		\begin{align*}
			\sem {\mu u~e}\{s/v\} &= \sem {\mu u~e}\{\sem 0/v\}\{s/v\} = \sem {\mu u~e}\{\sem 0\{s/v\}/v\} \\
			&= \sem {\mu u~e}\{\sem 0/v\}= \sem {\mu u~e}
		\end{align*}
		It follows that \(v\) is dead in \(\sem{\mu u~e}\). \qedhere
		\end{itemize}
	\end{itemize}
\end{proof}

\begin{lemma}\label{lem:compositionality sublemma 1}
	Let \(e,f \in \Exp\) and \(v \in V\).
	Assume that \(v\) is not free in \(f\), and that no free variable of \(f\) appears bound in \(e\).
	Then \(\sem{e}\{\sem f/v\} = \sem {e[f/v]}\).
\end{lemma}

\begin{proof}
	We proceed by induction on \(e\).
	\begin{itemize}
		\item For the variable case, let \(u \neq v\). There are cases to consider. 
		\begin{itemize}
			\item First, suppose \(e = v\). 
			We have \(
			\zeta(\sem{v}\{\sem{f}/v\})
			= \zeta(\sem{f})
			=\zeta (\sem{v[f/v]})
			\).
			\item Now assume \(e=u\). 
			Here, we have \(
			\zeta(\sem{u}\{\sem{f}/v\})
			= \zeta(\sem{u})
			= \zeta(\sem{u[f/v]})
			\).
		\end{itemize}
	
		\item For the inductive step, assume the result for \(e\) and consider the process term \(ae\).
		We have \(
		\zeta(\sem{ae}\{\sem{f}/v\})
		= [(a, \sem{e}\{\sem{f}/v\})]_{\Eq}
		= [(a, \sem{e[f/v]})]_{\Eq})
		= \zeta(\sem{ae[f/v]})
		\). 
		\item Now consider \(\sigma(e_1, \dots, e_n)\) for \(e_i \in \Exp\), \(i \le n\), and assume the result for \(e_1, \dots, e_n\).
		We have
		\begin{align*}
			\zeta(\sem{\sigma(e_1, \dots, e_n)}\{\sem f/v\})
			&= \sigma(\zeta(\sem{e_1}\{\sem f/v\}), \dots, \zeta(\sem{e_n}\{\sem f/v\})) \\
			&= \sigma(\zeta(\sem{e_1[f/v]}), \dots, \zeta(\sem{e_n[f/v]})) \\
			&= \zeta(\sem{\sigma(e_1, \dots, e_n)[f/v]})
		\end{align*}
		\item Finally, assume the result for \(e\) and consider \(\mu u~e\).
		Since no free variable of \(f\) is bound in \(\mu u~e\), \(u\) in particular is not free in \(f\).
		By \cref{lem:freeness and liveness}, \(u\) is therefore dead in \(\sem f\).
		Where \(\sem e = [p(v_1, \dots, v_n, (a_1,t_1), \dots, (a_m,t_m))]_{\Eq}\) and \(w_i = 0\) if \(v_i = u\) else \(w_i = v_i\), this leads to the computation
		\begin{align*}
			&\zeta(\sem{\mu u~e}\{\sem f/v\}) \\
			&= [p(w_1,\dots,w_n, (a_1, t_1\{\sem{\mu u~e}/u\}\{\sem f/v\}), \dots, (a_m, t_m\{\sem{\mu u~e}/u\}\{\sem f/v\}))]_{\Eq} \\
			&= [p(w_1,\dots,w_n, (a_1, t_1\{\sem f/v\}\{\sem{\mu u~e}\{\sem f/v\}/u\}), \dots, \\
			&\hspace{4cm} (a_m, t_m\{\sem f/v\}\{\sem{\mu u~e}\{\sem f/v\}/u\}))]_{\Eq} \tag{\cref{lem:composing behavioural substitutions}}\\
			&= \zeta(\sem{e}\{\sem f/v\})\{\sem{\mu u~e}\{\sem f/v\}/\!/u\}\\
			&= \zeta(\sem{e[f/v]})\{\sem{\mu u~e}\{\sem f/v\}/\!/u\} \tag{inductive hypothesis}
		\end{align*}
		It follows that, \(r = \sem{\mu u~e}\{\sem f/v\}\) satisfies the defining behavioural differential equation of \(\sem{\mu u~(e[f/v])}\).
		It follows that \(
			\sem{(\mu u~e)[f/v]} = \sem{\mu u~(e[f/v])} = \sem{e}\{f/v\} 
		\). \qedhere
	\end{itemize}
\end{proof}

The last two lemmas of this section show that guarded syntactic substitution at the level of \(\Exp\) plays well with guarded syntactic substitution at the level of \(B_M\Exp\).

\begin{lemma}\label{lem:guardingtheunguarded}
	Let \(v \in V\) and \(e,g \in \Exp\).
	Assume that no free variable of \(g\) appears bound in \(e\). 
	Then \(\epsilon(e)[g/\!/v] = \epsilon(e[g/\!/v])\).
\end{lemma}

\begin{proof}
	By induction on \(e\).
	\begin{itemize}
		\item If \(u \neq v\), then \(\epsilon(u)[g/\!/v] = [u]_{\Eq} = \epsilon(u[g/\!/v])\) since \(v\) does not appear in \(u\). 
		Otherwise, \(\epsilon(v)[g/\!/v] = [0]_{\Eq} = \epsilon(v[g/\!/v])\). 
		
		\item In the prefixing case, assume the result for \(e\) and simply observe that \((ae)[g/\!/v] = ae[g/v]\).
		Hence,
		\[
		\epsilon(ae)[g/\!/v] = [(a, e[g/v])]_{\Eq} = \epsilon(ae[g/v]) = \epsilon((ae)[g/\!/v])
		\] 
		
		\item Now assume the result for \(e_1, \dots e_n\).
		We have
		\begin{align*}
			\epsilon(\sigma(e_1,\dots, e_n))[g/\!/v]
			&= \sigma(\epsilon(e_1)[g/\!/v], \dots, \epsilon(e_n)[g/\!/v]) = \sigma(\epsilon(e_1[g/\!/v]), \dots, \epsilon(e_n[g/\!/v]))  \\
			&= \epsilon(\sigma(e_1[g/\!/v],\dots, e_n[g/\!/v])) = \epsilon(\sigma(e_1,\dots, e_n))[g/\!/v]
		\end{align*} 
		
		\item In the recursion step, assume the result for \(e\) and consider \(\mu u~e\).
		In particular, \(u\) is not free in \(g\). 
		It follows that \([g/\!/v][g/\!/v]=[g/\!/v]\), that \((\mu u~e)[g/\!/v] = \mu u~(e[g/\!/v])\), and that the operators \([g/\!/v]\) and \([\mu u~e[g/\!/v]/\!/u]\) commute.
		Thus, 
		\begin{align*}
			\epsilon(\mu u~e[g/\!/v]) 
			&= \epsilon(e[g/\!/v])[\mu u~e[g/\!/v]/\!/u] = \epsilon(e)[g/\!/v][\mu u~e[g/\!/v]/\!/u]	\\
			&= \epsilon(e)[g/\!/v][\mu u~e/\!/u][g/\!/v] = \epsilon(e)[\mu u~e/\!/u][g/\!/v][g/\!/v] = \epsilon(\mu u~e)[g/\!/v] \qedhere
		\end{align*}
	\end{itemize} 
\end{proof}

\begin{lemma}\label{lem:soundness sublemma 1}
	Let \(v \in V\) and \(e,g \in \Exp\).
	If \(v\) is guarded in \(e\) and no free variable of \(g\) appears bound in \(e\), then \(
	\epsilon(e[g/v]) = \epsilon (e)[g/\!/v]
	\).
\end{lemma}

\begin{proof}
	Suppose \(v\) is guarded in \(e\).
	We proceed by induction on the construction of \(e\).
	\begin{itemize}
		\item In the base case we only have to consider \(u \neq v\).
		Here, \(
		\epsilon(u[g/v]) 
		= \epsilon(u) 
		= [u]_{\Eq} 
		= [u[g/\!/v]]_{\Eq} 
		= \epsilon(u)[g/\!/v]
		\).
		
		\item In the inductive step, assume the result for \(e\) and consider \(ae\).
		We have
		\(
			\epsilon(ae[g/v]) 
			= [(a, e[g/v])]_{\Eq}
			= [(a, e)]_{\Eq} [g/\!/v]
			= \epsilon(e)[g/\!/v]
		\)
		as desired.
		
		\item In the branched case, assume the result for \(e_1, \dots, e_n\) and that \(v\) is guarded in \(e_1,\dots,e_n\).
		We have
		\begin{align*}
			\epsilon(\sigma(e_1,\dots,e_n)[g/v])
			&= \epsilon(\sigma(e_1[g/v],\dots,e_n[g/v])) 
			= \sigma(\epsilon(e_1[g/v]), \dots, \epsilon(e_n[g/v])) \\
			&= \sigma(\epsilon(e_1)[g/\!/v], \dots, \epsilon(e_n)[g/\!/v])) 
			= \sigma(\epsilon(e_1), \dots, \sigma(e_n))[g/\!/v]  \\
			&= \epsilon(\sigma(e_1,\dots,e_n))[g/\!/v]
		\end{align*} 
		
		\item Now assume the result for \(e\) and consider \(\mu u~e\).
		If \(v = u\), then we have \(\epsilon((\mu v~e)[g/v]) = \epsilon(\mu v~e) = \epsilon(e)[\mu v~e/\!/v]= \epsilon(e)[\mu v~e/\!/v][g/\!/v]= \epsilon(\mu v~e)[g/\!/v]\).
		Otherwise, assume \(v\) is guarded in \(e\) and compute
		\begin{align*}
			\epsilon(\mu u~e)[g/\!/v] 
			&= \epsilon(e)[\mu u~e/\!/u][g/\!/v]
			= \epsilon(e)[g/\!/v][\mu u~e[g/\!/v]/\!/u] \\
			&\stackrel{(\text{i.h.})} \epsilon(e[g/v])[\mu u~e[g/v]/\!/u] = \epsilon(\mu u~e[g/v])
		\end{align*}
		Here, \(e[g/\!/v] = e[g/v]\) precisely because \(v\) is guarded in \(e\). \qedhere
	\end{itemize}
\end{proof}

\section{Proofs from \cref{sec:final coalgebras}}
In this section, we aim to prove the following theorem, which states that the operational and denotational semantics coincide.

\compositionalitytheorem*

The proof requires the following lemma.

\begin{lemma}\label{lem:compositionality sublemma 2}
	Let \(p\in S^*(V+A\times \Exp)\), \(f\in \Exp\) and \(v\in V\). 
	Assume no free variable of \(g\) appears bound in any expression that appears in \(p\). 
	Then \[B_M(\sem-)([p]_{\Eq})\{\sem{g}/\!/v\} = B_M(\sem-)([p]_{\Eq}[g/\!/v])\]
\end{lemma}

\begin{proof}
	We proceed by induction on \(p\).
	
	\begin{itemize}
		\item Suppose \(p = u\).
		\begin{itemize}
			\item If \(u = v\), then we have \[
				B_M(\sem-)([v]_{\Eq})\{\sem g/\!/v\} = [v]_{\Eq}\{\sem g/\!/v\} = [0]_{\Eq} = B_M(\sem-)([0]_{\Eq}) = B_M(\sem-)([v]_{\Eq}[g/\!/v])
			\]
			\item If \(u \neq v\), then 
			\[
				 B_M(\sem-)([u]_{\Eq})\{\sem g/\!/v\}=[u]_{\Eq}\{\sem g/\!/v\} = [u]_{\Eq} = B_M(\sem-)([u]_{\Eq}) = B_M(\sem-)([u]_{\Eq}[g/\!/v])
				 \]
		\end{itemize}
		\item Now suppose \(p=(a,e)\) for some \(e \in \Exp\).
		We have that \(
		B_M(\sem-)([(a,e)]_{\Eq})\{\sem g/\!/v\}
		= [(a,\sem{e})]_{\Eq}\{\sem g/\!/v\}
		= [(a,\sem{e}\{\sem g/v\}]_{\Eq})
		\).
		By \cref{lem:compositionality sublemma 1},
		\begin{align*}
			B_M(\sem-)([(a,e)]_{\Eq})\{\sem g/\!/v\} 
			&= [(a,\sem{e}\{\sem g/v\}]_{\Eq}) 
			= [(a,\sem{e[g/v]})]_{\Eq} \\
			&= B_M(\sem -)([(a,e[g/v])]_{\Eq})
			= B_M(\sem -)([(a,e)]_{\Eq}[g/\!/v])
		\end{align*}
	
		\item Now assume the result for \(p_1, \dots, p_n \in S^*(V + A\times \Exp)\) and suppose \(p=\sigma(p_1, \dots, p_n)\). 
		We have 
		\begin{align*}
			&B_M(\sem-)([p]_{\Eq})\{\sem g/\!/v\} \\
			&= \sigma(B_M(\sem-)([p_1]_{\Eq}), \dots, B_M(\sem-)([p_n]_{\Eq}))\{\sem g/\!/v\} \\
			&= \sigma(B_M(\sem-)([p_1]_{\Eq})\{\sem g/\!/v\}, \dots, B_M(\sem-)([p_n]_{\Eq})\{\sem g/\!/v\}) \\
			&= \sigma(B_M(\sem-)([p_1]_{\Eq}[g/\!/v]), \dots, B_M(\sem-)([p_n]_{\Eq}[g/\!/v])) \\
			&= B_M(\sem-)([p]_{\Eq}[g/\!/v]) \qedhere
		\end{align*}
	\end{itemize}
\end{proof}

\begin{proof}[Proof of \cref{thm:compositionality}.]
	We prove the desired property by showing that \(\sem-\) is a \(B_M\)-coalgebra homomorphism between \((\Exp,\epsilon)\) and \((Z,\gamma)\).
	This amounts to showing that \(\zeta \circ \sem- = B_M(\sem-) \circ \epsilon\), which establishes the commutativity of the lower square in the diagram below.
	\[\begin{tikzpicture}
		\node (1) {\(\Sigma_MExp\)};
		\node[right=1.5cm of 1] (2) {\(\Sigma_MZ\)};
		\node[below=1cm of 1] (3) {\(Exp\)};
		\node[below=1cm of 2] (4) {\(Z\)};
		\node[below=1cm of 3] (5) {\(B_MExp\)};
		\node[right=1.5cm of 5] (6) {\(B_MZ\)};
		\draw (1) edge[->,dashed] node[above] {\(\Sigma_M\sem-\)} (2);
		\draw (1) edge[->] node[left] {\(\alpha\)} node[right] {\(\cong\)} (3);
		\draw (3) edge [->] node [left] {\(\epsilon\)} (5);
		\draw (3) edge [->, dashed] node [above] {\(\sem-\)} (4);
		\draw (4) edge [->] node [right] {\(\zeta\)} node [left] {\(\cong\)} (6);
		\draw (5) edge [->, dashed] node [above] {\(B_M\sem - \)} (6);
		\draw (2) edge [->] node [right] {\(\gamma\)} (4);
	\end{tikzpicture}\]
	To this end, we show \(\zeta(\sem e) = B_M(\sem-)(\epsilon(e))\) by induction on \(e\).
	\begin{itemize}
		\item In the base case, \(e = v\), and
		\[
		\zeta(\sem{v}) = [v]_{\Eq} = B_M(\sem-)([v]_{\Eq}) = B_M(\sem-)(\epsilon(v))
		\]
	
		\item Now assume the result for \(f\) and let \(e = af\). 
		We have 
		\[
		\zeta(\sem{af}) = [(a,\sem f)]_{\Eq} = B_M(\sem-)([(a,f)]_{\Eq}) = B_M(\sem-)(\epsilon(af))	
		\]
		
		\item Next, assume the result for \(e_1, \dots, e_n\) and let \(e=\sigma(e_1, \dots, e_n)\). 
		We have
		\begin{align*}
			\zeta(\sem{e}) 
			&= \sigma(\zeta(\sem {e_1}), \dots, \zeta(\sem{e_n}))
			= \sigma(B_M(\sem-)(\epsilon(e_1)), \dots, B_M(\sem-)(\epsilon(e_n)))\\
			&= B_M(\sem-)(\sigma(\epsilon(e_1), \dots, \epsilon(e_n)))
			= B_M(\sem-)(\epsilon(e))
		\end{align*}
		
		\item Now assume the result for \(f\) and let \(e = \mu v~f \). 
		We have 
		\begin{align*}
			\zeta(\sem{\mu v~f}) 
			&= \zeta(\sem f)\{\sem{\mu v~f}/\!/v\} 
			= B(\sem-)(\epsilon(f)) \{\sem{\mu v~f}/\!/v\} \\
			&= B(\sem-)(\epsilon(f)[\mu v~f/\!/v]) \tag{\cref{lem:compositionality sublemma 2}} \\
			&= B(\sem-)(\epsilon(\mu v~f)) \qedhere
		\end{align*}
	\end{itemize}
\end{proof}

\section{Proofs from \cref{sec:soundness and completeness}}

\soundnesstheorem*

The proof of this lemma requires the following lemma.

\begin{lemma}\label{lem:soundness sublemma 2}
	Let \(v \in V\) and \(g_1,g_2 \in \Exp\).
	If \(g_1 \equiv g_2\), then for any term \(p \in S^*(V + A\times \Exp)\) we have \[
	B_M([-]_\equiv)(p[g_1/\!/v]) = B_M([-]_\equiv)(p[g_2/\!/v])
	\]
\end{lemma}

\begin{proof}
	By induction on the construction of \(p\). 
	\begin{itemize}
		\item Suppose \(p = u\).
		If \(u \neq v\), then we have \(u[g_1/\!/v] = u = u[g_2/\!/v]\), because \(v\) is not free in \(u\).
		On the other hand, if \(u = v\), then \(v[g_i/\!/v] = 0\) for \(i = 1,2\) by definition.
		
		\item Now let \(p = (a, e)\). 
		We have \[
			B_M([-]_\equiv)((a,e)[g_i/\!/v])
			= B_M([-]_\equiv)((a,e[g_i/v]))
			= [(a,[e[g_i/v]]_\equiv)]_{\Eq}
		\]
		for \(i = 1,2\).
		Since \(\equiv\) is a congruence, \(e[g_1/v] \equiv e[g_2/v]\), so indeed \((a,[e[g_1/v]]_\equiv) = (a,[e[g_2/v]]_\equiv)\) as desired. 
		
		\item In the branched case, assume the identity for \(p_1,\dots, p_n\) and compute 
		\begin{align*}
			& B_M([-]_\equiv)(\epsilon(\sigma(p_1,\dots, p_n))[g_i/\!/v]) \\
			& = B_M([-]_\equiv)(\sigma(p_1[g_i/\!/v],\dots, p_n[g_i/\!/v])) \\
			&= \sigma(B_M([-]_\equiv)(p_1[g_i/\!/v]),\dots, B_M([-]_\equiv)(p_n[g_i/\!/v])) \\
			&= \sigma(B_M([-]_\equiv)(p_1[g_i/\!/v]),\dots, B_M([-]_\equiv)(p_n[g_i/\!/v])) 
		\end{align*}
		for \(i = 1,2\). 
		The desired identity now follows from the induction hypothesis, which states here that \(B_M([-]_\equiv)(p_j[g_i/\!/v]) = B_M([-]_\equiv)(p_j[g_2/\!/v]) \) for \(j \le n\). \qedhere
	\end{itemize}
\end{proof}

Next we prove \cref{thm:soundness} using a technique that is similar in style to others that currently exist in the literature, for example in \cite{jacobs2006bialgebraic} and \cite{silva2010kleene}.

\begin{proof}(of \cref{thm:soundness})
	At present, we are given the following three maps:
	\[\label{eq:soundness square}
		\begin{tikzpicture}
			\node (1) {\(\Exp\)};
			\node[right=3cm of 1] (2) {\(\Expm\)};
			\node[below=1cm of 1] (3) {\(B_M\Exp\)};
			\node[below=1cm of 2] (4) {\(B_M\Expm\)};
			\node[below=0.25cm of 1, xshift=1.8cm] (5) {(\(*\))};
			\draw (1) edge[->] node[above] {\([-]_\equiv\)} (2);
			\draw (1) edge[->] node[left] {\(\epsilon\)} (3);
			\draw (3) edge[->] node[below] {\(B_M([-]_\equiv)\)} (4);
		\end{tikzpicture}\]
	We will show that there is a map \(\bar \epsilon : \Expm \to B_M \Expm\) such that the resulting square commutes. 
	Since \([-]_\equiv\) is surjective, there is at most one such map, because the existence of \(\bar \epsilon\) is equivalent to the statement that \(k := B_M([-]_\equiv)\circ \epsilon\) is constant on \(\equiv\)-equivalence classes of terms.
	Thus, it suffices to show that if \(e \equiv f\), then \(k(e) = k(f)\). 
	We proceed by induction on the length of the derivation of \(e \equiv f\).
	
	In the base case, \(e \equiv f\) is an instance of one of the axioms. 
	That is, \(e \equiv f\) either (1) is an instance of an axiom of \(\Eq\), (2) is an instance of (\acro R1), or (3) is an instance of (\acro R2).
	\begin{itemize}
		\item[(1)] Suppose for the sake of argument that \(\Eq \subseteq S^*X \times S^*X\) and \((t,s) \in \Eq\), and let \(\nu : X \to \Exp\) lift to \(\nu^\# : S^* X \to \Exp\). 
		If we define the map \(h : \Exp \to S^*(V + A\times \Exp)\) inductively as
		\[
		\begin{aligned}
			h(v) &= v \\
			h(ae) &= (a, e)\\
		\end{aligned}
		\qquad\qquad
		\begin{aligned}
			h(\sigma(e_1,\dots,e_n)) &= \sigma(h(e_1), \dots, h(e_n)) \\
			h(\mu v~ e) &= h(e)[\mu v~e/\!/v]
		\end{aligned}
		\]
		then \(\epsilon = [-]_{\Eq} \circ h\).
		Hence, if \(e = \nu^\#(t)\) and \(f = \nu^\#(s)\), then 
		\begin{align*}
			\epsilon(e) &= \epsilon\circ \nu^\#(t) \tag{def. \(e\)}\\
						&= [-]_\equiv \circ h\circ \nu^\#(t)  \tag{def. \(\epsilon\)}\\
						&= ([-]_\equiv \circ h\circ \nu)^\#(t) \tag{univ. property of \((-)^\#\)}\\
						&= ([-]_\equiv \circ h\circ \nu)^\#(s) \tag{\((\Expm, \hat \alpha)\) satisfies \(\Eq\)}\\
						&= [-]_\equiv \circ h\circ \nu^\#(s) \tag{univ. property of \((-)^\#\)}\\
						&= \epsilon(f) \tag{def. \(f\)}
		\end{align*}
		where \(\hat \alpha : \Sigma_M \Expm \to \Expm\) is the quotient algebra of \((\Exp, \alpha)\) by the congruence \(\equiv\). 
		It follows that if \(e \equiv f\) is an axiom of \(\Eq\), then \(\epsilon(e) = \epsilon(f)\), and therefore \(k(e) = k(f)\). 
		
		\item[(2)] Next we consider the equation \(\mu v~e \equiv e[\mu v~e/\!/v]\) of (\acro R1). 
		By definition, \(\epsilon(\mu v~e) = \epsilon(e)[\mu v~e/\!/v]\), and by \cref{lem:guardingtheunguarded} we know that \(\epsilon(e)[\mu v~e/\!/v] = \epsilon(e[\mu v~e/\!/v])\). 
		It immediately follows that \(k(\mu v~e) = k(e[\mu v~e/\!/v])\). 
		
		\item[(3)] Thirdly, we consider the equation \(\mu v~e \equiv \mu w~e[w/v]\) of (\acro R2), in which \(w\) does not appear freely in \(e\). 
		This follows from \cref{lem:soundness sublemma 2}: We have
		\begin{align*}
			B_M([-]_{\equiv})(\epsilon(\mu v~e)) 
			&= B_M([-]_{\equiv})(\epsilon(e)[\mu v~e/\!/v]) \\
			&= B_M([-]_{\equiv})(\epsilon(e)[\mu w~e[w/v]/\!/v]) \tag{\cref{lem:soundness sublemma 2}}\\
			&= B_M([-]_{\equiv})(\epsilon(e[w/v])[\mu w~e[w/v]/\!/w]) \\
			&= B_M([-]_{\equiv})(\epsilon(\mu w~e[w/v]))
		\end{align*}
	\end{itemize}

	For the inductive step, we assume that the proof of \(e \equiv f\) ends either (1) with a deduction rule from equational logic, (2) ends with one of the congruence-generating rules
	\[
		\prftree[r]{(\(S\)-cong)}{(\forall i \le n)~e_i \equiv f_i}{\sigma(e_1,\dots,e_n) \equiv \sigma(f_1,\dots,f_n)}
		\\
		\prftree[r]{(\(A\)-cong)}{e \equiv f}{ae \equiv af}
	\]
	or (3) ends with the rule (\acro R3).
	\begin{itemize}
		\item[(1)] The inference rules of equational logic respect identities across function application, so this case is trivial.
		
		\item[(2)] If the last step is of the proof is (\(S\)-cong), then use the fact that \(\epsilon = [-]_{\Eq}\circ h^\#\) as in the first step of the base case. 
		If the last step is (\(A\)-cong), then observe that 
		\begin{align*}
			B_M([-]_{\equiv})(\epsilon(ae))
			&= B_M([-]_{\equiv})([(a,e)]_{\Eq}) = [(a, [e]_\equiv)]_{\Eq} = [(a, [f]_\equiv)]_{\Eq} \\
			&= B_M([-]_{\equiv})([(a, f)]_{\Eq}) = B_M([-]_{\equiv})(\epsilon(af)) \\
		\end{align*}
		
		\item[(3)] Now suppose the last rule is (\acro R3), and assume that \(v\) is guarded in \(e\).
		We have 
		\begin{align*}
			B_M([-]_\equiv)(\epsilon(\mu v~e))  
			&= B_M([-]_\equiv)(\epsilon(e)[\mu v~e/\!/v]) \tag{def. of \(\epsilon\)}\\
			&= B_M([-]_\equiv)(\epsilon(e)[g/\!/v]) \tag{\cref{lem:soundness sublemma 2}}\\
			&= B_M([-]_\equiv)(\epsilon(e[g/v])) \tag{\cref{lem:soundness sublemma 1}}\\ 
			&= B_M([-]_\equiv)(\epsilon(g)) \tag{induction hypothesis}\\
		\end{align*}
	\end{itemize}
	It follows that \(\equiv\), which is equal to \(\ker([-]_\equiv)\), is contained in \(\ker(B([-]_\equiv)\circ\epsilon)\). 
	Thus, there is a unique map \(\bar \epsilon : \Expm \to B_M(\Expm)\) such that the resulting square (\(*\)) commutes. 
\end{proof}

\subsection{Completeness}

In the following lemma, we assume that \(S\) only has operations of finite arity (this is Assumption~\ref{asm:finite arities}).

\localfinitenesscondition*

\begin{proof}
	Given an arbitrary \(e \in \Exp\), we will explicitly construct a subcoalgebra of \((\Exp, \epsilon)\) that has a finite set of states that includes \(e\).
	To this end, define \(U : \Exp \to \P(\Exp)\) by
	\[
		U(v) = \{v\}
		\qquad 
		U(ae) = \{ae\} \cup U(e)
		\qquad 
		U(\sigma(e_1, \dots, e_n)) = \{\sigma(e_1, \dots, e_n)\} \cup \bigcup_{i < n} U(e_i) 
	\]
	\vspace*{-1em}
	\[
		U(\mu v~e) = \{\mu v~e\} \cup U(e)[\mu v~e/\!/v] := \{\mu v~e\} \cup \{f[\mu v~e/\!/v] \mid f \in U(e)\}
	\]
	Note that \(e \in U(e)\) for all \(e \in \Exp\), and that \(U(e)\) is finite. 
	We begin with following claim, which says that the derivatives of \(e\) can be given in terms of expressions from \(U(e)\): 
	for any \(e \in \Exp\), there is a representative \(S\)-term \(p \in \epsilon(e)\) such that \(p \in S^*(V + A \times U(e))\). 
	This can be seen by induction on \(e\). 
	The only interesting case is the inductive step \(\mu v~e\), in which case we let \(p \in \epsilon(e)\) and note that \(p[\mu v~e/\!/v]\) is a representative of \(\epsilon(\mu v~e)\) in \(S^*(V + A \times U(\mu v~e))\).
	
	To finish the proof of the lemma, fix an \(e \in \Exp\) and define a sequence of sets beginning with \(U_0 = \{e\}\) and proceeding with 
	\[
	 	U_{n+1} = U_n \cup \bigcup_{e_0 \in U_n}\{g \mid (\exists a \in A)(\exists p \in \epsilon(e_0)\cap S^*(V + A\times U(e_0)))~\text{\((a,g)\) appears in \(p\)}\}
	\] 
	Then \(U_0 \subseteq U_1 \subseteq \cdots \subseteq U(e)\), and the latter set is finite. 
	Hence \(U := \bigcup U_n\) is finite and contained in \(U(e)\). 
	We define a coalgebra structure \(\epsilon_U : U \to B_MU\) by taking \(\epsilon_U(e) = [p]_{\Eq}\) where if \(e \in U_n\), then \(p\) is a representative of \(\epsilon(e)\) in \(S^*(V + A\times U_{n+1})\).
	Since \(S^*(V + A \times U_{n+1}) \subseteq S^*(V + A \times U)\), this defines a \(B_M\)-coalgebra structure on \(U\).
	Where \(\iota : U \hookrightarrow \Exp\), we have \(\epsilon(\iota(e)) = \epsilon(e) = B_M(\iota)\circ\epsilon_U(e)\). 
	Thus, \((U,\epsilon_U)\) is a finite subcoalgebra of \((\Exp,\epsilon)\) containing \(e\). 
\end{proof}

\solutionsarehomomorphisms*

\begin{proof}
	We begin by observing that \(\bar \epsilon : \Expm \to B_M\Expm\) is a bijection. 
	Indeed, the map \((-)^\heartsuit : B_M\Exp \to \Expm\) defined\[
		[v]_{\Eq}^\heartsuit = [v]_\equiv 
		\qquad 
		[(a,e)]_{\Eq}^\heartsuit = [ae]_\equiv
		\qquad
		[\sigma(p_1,\dots,p_n)]_{\Eq}^\heartsuit = \sigma([p_1]_{\Eq}^\heartsuit, \dots, [p_n]_{\Eq}^\heartsuit)
	\] 
	is its inverse:
	Clearly \(\bar\epsilon([p]_\Eq^\heartsuit) = [p]_\Eq\) for any \(p \in S^*(V + A\times \Exp)\), so it suffices to see that \(\bar\epsilon([e]_{\equiv})^\heartsuit = [e]_{\equiv}\) for all \(e \in \Exp\).
	This can be done by induction on \(e\), but the only tricky case is \(\mu v~e\). 
	For this case, observe that
	\begin{align*}
		\bar\epsilon([\mu v~e]_\equiv)^\heartsuit 
		&= (B_M([-]_\equiv)(\epsilon(e)[\mu v~e/\!/v]))^\heartsuit
		= (B_M([-]_\equiv)(\epsilon(e[\mu v~e/\!/v])))^\heartsuit \\
		&= (\bar\epsilon([e[\mu v~e/\!/v]]_\equiv))^\heartsuit
		= (\bar\epsilon([e]_\equiv)[[\mu v~e]_\equiv/\!/v])^\heartsuit
		= \bar\epsilon([e]_\equiv)^\heartsuit [[\mu v~e]_\equiv/\!/v]  \\
		&\stackrel{\text{(i.h.)}}= [e[\mu v~e/\!/v]]_\equiv
		= [\mu v~e]_\equiv
	\end{align*}
	where the fifth equality is the induction hypothesis and the last is (\acro {R1}).
	We have also made use of a lifting of syntactic substitution to \(B_M(\Expm)\) that is defined in the usual way, as well as the fact that this lifted syntactic substitution commutes with \((-)^\heartsuit\), which can be proven by induction on terms \(S^*(V + A\times (\Expm))\). 
	
	Now let \(\{x = p_x^\dagger\}_{x \in X}\) be the system of equations associated with the coalgebra \((X,\beta)\).
	Observe that for any \(x,y \in X\), if \(y\) appears in \(p_x\), then it is guarded in \(p_x^\dagger\). 
	This means that \(\phi : X \to \Exp\) is a solution to \(\{x = p_x^\dagger\}_{x \in X}\) if and only if \(\phi(x) \equiv p_x^\dagger[\phi(y)/\!/y]_{y \in X}\).
	Now, if \(\beta(x) = [p_x]_{\Eq}\), we see that
	\begin{align*}
	(B_M([-]_\equiv \circ \phi)(\beta(x)))^\heartsuit 
	&= (B_M([-]_\equiv) \circ B_M(\phi)([p_x]_{\Eq}))^\heartsuit \\
	&= (B_M([-]_\equiv)([p_x]_{\Eq}[\phi(y)/\!/y]_{y \in X}))^\heartsuit \\
	&= ([p_x]_{\Eq}[ [\phi(y)]_\equiv/\!/y ]_{y \in X})^\heartsuit \\
	&= [p_x^\dagger[\phi(y)/\!/y]_{y \in X}]_\equiv 
	\end{align*}
	Thus, \(\phi\) is a solution to the system \(\{x=p_x^\dagger\}_{x \in X}\) if and only if  
	\begin{equation}\label{eq:dagger bout to get pwned}
		[-]_\equiv \circ \phi(x) = (-)^\heartsuit \circ B_M([-]_\equiv \circ \phi) \circ \beta (x) 
	\end{equation}
	for every \(x \in X\).
	The maps \((-)^\heartsuit\) and \(\bar\epsilon\) are inverse to one another, so \cref{eq:dagger bout to get pwned} is equivalent to the identity \(\bar \epsilon \circ [-]_\equiv \circ \phi = B_M([-]_\equiv\circ\phi)\circ\beta\).
	This identity is the defining property of a coalgebra homomorphism of the form \([-]_\equiv~\phi\).
\end{proof}

\milnerslemma*

The following proof is a recreation of the one that appears under \cite[Theorem 5.7]{milner1984complete} with the more general context of our paper in mind.
Remarkably, the essential details of the proof remain unchanged despite the jump in the level of abstraction between the two results. 

\begin{proof}
	Let \(\{x_i = e_i\}_{i \le n}\) be a guarded system of equations.
	We proceed by induction on \(n\).
	In the base case, \(n = 1\).
	This case is straight-forward because \(\phi(x_1) := \mu x_1~e_1\) is its unique solution up to \(\equiv\) by (\acro {R3}). 
	
	Now assume that every system of strictly fewer than \(n\) guarded equations has a unique solution up to \(\equiv\).
	Define 
	\begin{align*}
		f_n &= \mu x_n~e_n
		&\text{and}&&
		f_i &= e_i [f_n/x_n]
	\end{align*}
	for each \(i < n\).
	Since \(x_1,\dots, x_n\) are guarded in \(e_i\) for \(i \le n\), the system \(\{x_i = f_i\}_{i < n}\) is also guarded, and \(x_1,\dots, x_{n-1}\) do not appear freely in any \(f_i\) for \(i < n\).
	By the induction hypothesis, \(\{x_i = f_i\}_{i < n}\) has a unique solution \(\psi : \{x_1, \dots, x_{n-1}\} \to \Exp\) up to \(\equiv\).
	Let \(g_i = \psi(x_i)\) for \(i < n\) and note that \(x_n\) is not free and does not appear bound in any of \(f_1,\dots, f_{n-1}\) by construction.
	Now, take \(g_n = f_n[g_1/x_1, \dots, g_{n-1}/x_{n-1}]\).
	Then \(\phi(x_i) := g_i\) for \(i \le n\) is indeed a solution of the desired form, since
	\begin{align*}
		g_n &= f_n[g_1/x_1, \dots, g_{n-1}/x_{n-1}] = (\mu x_n~e_n)[g_1/x_1, \dots, g_{n-1}/x_{n-1}] \\
		&= \mu x_n~(e_n[g_1/x_1, \dots, g_{n-1}/x_{n-1}]) \tag{\(x_n\) not free in \(g_i\)}\\
		&\equiv e_n[g_1/x_1, \dots, g_{n-1}/x_{n-1}][g_n/x_n] \tag*{(\acro R1),(\(x_n\) guarded in \(g_n\))}\\
		&= e_n[g_1/x_1, \dots, g_{n-1}/x_{n-1}, g_n/x_n] \tag{\(x_n\) not free in \(g_i\)}
	\end{align*} 
	and for any \(i < n\),
	\begin{align*}
		g_i &\equiv f_i [g_1/x_1, \dots, g_{n-1}/x_{n-1}] \\
		&= e_i [f_n/x_n][g_1/x_1, \dots, g_{n-1}/x_{n-1}] = e_i [g_1/x_1, \dots, g_{n-1}/x_{n-1}][f_n[g_1/x_1, \dots, g_{n-1}/x_{n-1}]/x_n] \\
		&= e_i [g_1/x_1, \dots, g_{n-1}/x_{n-1}, f_n[g_1/x_1, \dots, g_{n-1}/x_{n-1}]/x_n] = e_i [g_1/x_1, \dots, g_{n-1}/x_{n-1}, g_n/x_n]
	\end{align*}
	since \(x_n\) is not free in \(g_i\) for any \(i < n\).
	
	To see that the solution is unique, let \(\theta(x_i) = h_i\) for \(i \le n\) be any other solution to \(\{x_i = e_i\}_{i \le n}\).
	Then in particular, \(h_n \equiv e_n[h_1/x_1, \dots, h_{n-1}/x_{n-1}, h_n/x_n] = e_n[h_1/x_1, \dots, h_{n-1}/x_{n-1}][h_n/x_n]\) since \(x_n\) is not free in any \(h_1,\dots,h_{n-1}\).
	This means that \(h_n \equiv \mu x_n~(e_n[h_1/x_1, \dots, h_{n-1}/x_{n-1}])\) by (\acro R3) and guardedness of \(x_n\) in \(e_n\).
	Since \(x_n\) is not free in \(h_i\) for any \(i < n\),
	\begin{align*}
	\mu x_n~(e_n[h_1/x_1, \dots, h_{n-1}/x_{n-1}]) 
	&= (\mu x_n~e_n)[h_1/x_1, \dots, h_{n-1}/x_{n-1}] \\
	&= f_n[h_1/x_1, \dots, h_{n-1}/x_{n-1}]
	\end{align*}
	This makes the restriction of \(\theta\) to \(x_1,\dots,x_{n-1}\) a solution to \(\{x_i = f_i\}_{i < n}\).
	By the induction hypothesis, there is only one such solution, so \(h_i \equiv g_i\) for each \(i < n\).
	It follows that that\[
	h_n \equiv \mu x_n~(e_n[h_1/x_1, \dots, h_{n-1}/x_{n-1}])
	\equiv \mu x_n~(e_n[g_1/x_1, \dots, g_{n-1}/x_{n-1}])
	= g_n
	\] 
	via the congruence laws.
	Hence, \(h_n \equiv g_n\), and overall \(\theta \equiv \phi\). 
\end{proof}

\section{Proofs from \cref{sec:star fragments}}\label{sec:unguarded unravelling}

In this appendix, we are concerned with a particular statement made near the end of \cref{sec:star fragments}: that the theory \(\Eq^*\) is equipotent to the axiomatisations found in the literature, in the particular cases of \(\Eq = \mathsf{SL}\) and \(\Eq= \mathsf{GS}\). 
This requires us to show that the unrestricted equations of the form \(e^{(\sigma)} = ee^{(\sigma)} +_\sigma \skiptt\) are derivable from \(\Eq^*\) in each of these two cases.

First let us consider the case of Milner's star fragment.
Write \(e \to \checkmark\) if \(\checkmark \in \ell(e)\).
Define the operator \(\partial : \SExp \to \SExp\) by induction as
\[
	\partial 0 = 0 = \partial \skiptt 
	\qquad 
	\partial a = a
	\qquad \partial (e + f) = \partial e + \partial f
	\qquad
	\partial(ef) = \begin{cases}
		\partial e f + \partial f &e \to \checkmark \\
		\partial e f & e \not \to \checkmark
	\end{cases}
\]
Note that \(\partial e\) is guarded for all \(e\). 
We have the following.

\begin{lemma}
	For any \(e \in \SExp\), 
	\begin{itemize}[leftmargin=1cm]
		\item [(i)] If \(e \to \checkmark\), then \(\mathsf{SL}^* \vdash e = \partial e + 1\), else \(\mathsf{SL}^* \vdash e = \partial e\).
		\item [(ii)] \(e^* = (\partial e)^*\)
	\end{itemize}
\end{lemma}

In this proofs that follow, we write simply \(e = f\) in place of \(\mathsf{SL}^* \vdash e = f\).

\begin{proof}	
	Statement (ii) follows directly from (i) and (\acro S4).
	We prove statement (i) by induction on \(e\).
	The base cases hold by definition, so we proceed to the inductive step and assume (i) holds for \(e\) and \(f\).
	\begin{itemize}
		\item If \(e + f \to \checkmark\), then either \(e \to \checkmark\) or \(f \to \checkmark\) and the induction hypothesis directly applies. 
		For example, if \(e \to \checkmark\) and \(f \to \checkmark\), then \(e + f = (\partial e + \skiptt) + (\partial f + \skiptt) = \partial e + \partial f + \skiptt = \partial (e + f) + \skiptt\).

		\item In case \(e + f \not\to \checkmark\), we simply have \(e + f = \partial e + \partial f = \partial (e + f)\). 

		\item If \(ef \to \checkmark\), then \(e \to \checkmark\) and \(f \to \checkmark\), and we have \(ef = (\partial e + \skiptt)f = \partial ef + f = \partial e f + \partial f + \skiptt = \partial (ef) + \skiptt\).

		\item If \(ef \not \to \checkmark\), then either \(e \not\to \checkmark\) or \(f \not\to \checkmark\).
		\begin{itemize}
			\item In the first case, \(ef = \partial ef = \partial (ef)\).
			\item And in case \(e \to \checkmark\) but \(f \not\to\checkmark\), \(ef = (\partial e + \skiptt)f = \partial ef + \partial f = \partial (ef).\)
		\end{itemize}
		\item Since \(e^* \to \checkmark\), we need to see that \(e^* = \partial(e^*) + \skiptt\).
		There are two cases to consider:
		\begin{itemize}
			\item If \(e \to \checkmark\), then 
			\(e^* = (\partial e + 1)^* = (\partial e + 0)^* = (\partial e)^* = \partial e(\partial e)^* + \skiptt = \partial ee^* + \skiptt = \partial (e^*) + \skiptt\).
			\item  Otherwise, we have \(e^* = (\partial e)^* \stackrel{(\mathsf S5)}{=} \partial e(\partial e)^* + \skiptt = \partial ee^* + \skiptt = \partial (e^*) + \skiptt\). \qedhere
		\end{itemize}
	\end{itemize}
\end{proof}

\begin{theorem}
	For any \(e \in \Exp\), \(\mathsf{SL}^* \vdash e^* = ee^* + 1\). 
\end{theorem}

\begin{proof}
	The statement is equivalent to (\acro S5) if \(e \not\to \checkmark\), so it suffices to show the case where \(e \to \checkmark\).
	Since \(ee^* + 1 \to \checkmark\),
	\begin{align*}
		ee^* + 1 &= \partial (ee^* + 1) + 1 = \partial (ee^*) + 1\\
		&= \partial ee^* + \partial (e^*) + 1 = \partial e e^* + \partial e e^* + 1 = \partial ee^* + 1 = e^* \qedhere
	\end{align*}
\end{proof}

Next we consider the case \(\Eq = \mathsf {GS}\).
Write \(e \Rightarrow b\) if \(b = \{\xi \in \At \mid \ell(e)(\xi) = \checkmark\}\).
We follow in the footsteps of the previous proof and define the operator \(\partial : \SExp \to \SExp\) inductively by
\[
	\partial 0 = 0 = \partial \skiptt 
	\qquad 
	\partial a = a
	\qquad 
	\partial (e +_c f) = \partial e +_c \partial f
\]
and if \(e \Rightarrow b\), \[
	\partial (ef) = \partial f +_b \partial ef
	\qquad 
	\partial (e^{(c)}) = 0 +_b \partial e e^{(c)}
\]
Note that \(\partial e \Rightarrow \emptyset\) for all \(e \in \SExp\).

\begin{lemma}
	For any \(e \in \SExp\), if \(e \Rightarrow b\), then 
	\begin{itemize}
		\item [(i)] \(\mathsf{GS}^* \vdash e = \skiptt +_b \partial e\)
		\item [(ii)] \(\mathsf{GS}^* \vdash e^{(c)} = (0 +_b \partial e)^{(c)}\)
	\end{itemize}
\end{lemma}

Again, we will write \(e = f\) in place of \(\mathsf{GS}^* \vdash e=f\) in the following proofs. 

\begin{proof}
	Again, (ii) follows directly from (i), and we prove (i) by induction on \(e\). 
	The base cases hold by definition, so it suffices to assume (i) for \(e\) and \(f\).
	Let \(e \Rightarrow b_1\) and \(f \Rightarrow b_2\). 
	Then since \(e +_c f \Rightarrow b_1c \cup b_2\bar c\), setting \(b = b_1c \cup b_2\bar c\) we derive\[
		e +_c f = (\skiptt +_{b_1} \partial e) +_c (\skiptt +_{b_2} \partial f) 
		= \skiptt +_{b_1c \cup b_2\bar c} (\partial e +_c \partial f) 
		= \skiptt +_b \partial (e +_c f)
	\]
	Next, if \(b = b_1b_2\), then \(ef \Rightarrow b\) and we derive
	\[
		ef = (\skiptt +_{b_1} \partial e)f 
		= (\skiptt +_{b_2} \partial f) +_{b_1} \partial ef
		= \skiptt +_{b_1b_2} (\partial f +_{b_1} \partial ef)
		= \skiptt +_{b} \partial (ef)
	\]
	Finally, \(e^{(c)} \Rightarrow \bar c\), so  
	\begin{align*}
		e^{(c)} 
		&= (\skiptt +_{b_1} \partial e)^{(c)}= (0 +_{b_1} \partial e)^{(c)}  \\
		&= (0 +_{b_1} \partial e)(0 +_{b_1} \partial e)^{(c)} +_c \skiptt  \tag {\acro S5}\\
		&= (0 +_{b_1} \partial ee^{(c)}) +_c \skiptt = \skiptt +_{\bar c} (0 +_{b_1} \partial ee^{(c)}) = \skiptt +_{\bar c} \partial (e^{(c)})\qedhere
	\end{align*}
\end{proof}

\begin{theorem}
	For any \(e \in \SExp\) and \(b \subseteq \At\), \(\mathsf{GS}^* \vdash e^{(b)} = ee^{(b)} +_b \skiptt\).
\end{theorem}

\begin{proof}
	Suppose \(e \Rightarrow c\). 
	Using the previous lemma, derive
	\begin{align*}
		ee^{(b)} +_b \skiptt 
		&= \skiptt +_{\bar b} \partial (ee^{(b)} +_b \skiptt ) = \skiptt +_{\bar b} (\partial (ee^{(b)}) +_b 0 ) = \skiptt +_{\bar b} \partial (ee^{(b)})\\
		&= \skiptt +_{\bar b} (\partial e^{(b)} +_c \partial ee^{(b)})= \skiptt +_{\bar b} ((0 +_c \partial e e^{(b)}) +_c \partial ee^{(b)}) \\
		&= \skiptt +_{\bar b} (0 +_c \partial ee^{(b)}) = \skiptt +_{\bar b} \partial e^{(b)} = e^{(b)}\qedhere
	\end{align*}
\end{proof}

\end{document}